\documentclass{CSML}

\usepackage{lastpage}

\def\dOi{13(3:21)2017}
\lmcsheading%
{\dOi}
{1--39}
{}
{}
{Mar.~15, 2016}
{Sep.~12, 2017}
{}

	\usepackage[utf8]{inputenc} 
	\usepackage[all,cmtip]{xy}
	\usepackage{hyperref, xspace, amssymb, cleveref, stmaryrd, tikz, pgfplots, wrapfig,caption}
        \hypersetup{hidelinks}
	\pgfplotsset{compat=1.12}
	\creflabelformat{enumi}{#2(#1)#3} 
	\crefname{enumi}{}{}
	
	\DeclareMathOperator{\NN}{\mathbb N}
	\DeclareMathOperator{\ZZ}{\mathbb Z}
	 
	\DeclareMathOperator{\RR}{\mathbb R}
	\DeclareMathOperator{\DD}{\mathbb D}

	\DeclareMathOperator{\C}{\mathcal C}
	\DeclareMathOperator{\B}{\mathcal B}
	\DeclareMathOperator{\M}{\mathcal M}
	\newcommand{\PP}{\mathcal P}
	\newcommand{\NP}{\mathcal{NP}}
	\newcommand{\FP}{\mathcal{FP}}

	\newcommand{\capa}[1]{\mathrm{cap}\left(#1\right)}
	\newcommand{\size}[1]{\left|#1\right|}
	\newcommand{\bsize}[1]{\big|#1\big|}
	\newcommand{\aaT}{Arzelà-Ascoli\xspace}
	\newcommand{\fkT}{Fréchet-Kolmogorov\xspace}

	\newcommand{\dd}{\mathrm d}
	\newcommand{\lb}{\mathrm{lb}}
	\newcommand{\length}[1]{\left|#1\right|}
	
	\newcommand{\flength}[1]{\left|#1\right|}
	\newcommand{\sdzero}{\textup{\texttt{0}}}
	\newcommand{\sdone}{\textup{\texttt{1}}}
	\newcommand{\binalbe}{\ensuremath{\left\{\sdzero,\sdone\right\}}\xspace}
	\newcommand{\albe}{\Sigma}
	\newcommand{\demph}[1]{\textbf{#1}}
	
	\newcommand{\eval}{\mathrm{eval}}
	
	\newcommand{\convo}{\star}
	\newcommand{\cell}[2]{\left[#1\right]_{#2}}
	\newcommand{\LL}[1]{{L}^{#1}}
	\newcommand{\LLL}[2]{\LL{#1}\!\left(#2\right)}
	\newcommand{\Lp}{\LL{p}}
	\newcommand{\Lone}{\LL{1}}
	\newcommand{\xip}{\xi_{p}}
	
	\newcommand{\xic}{\xi_{C}}
	\newcommand{\xis}{\xi_s}
	\newcommand{\dom}{\mathrm{dom}}
	\newcommand{\grad}{\nabla}
	
	\newcommand{\p}{\mathcal{P}}
	\newcommand{\np}{\mathcal{NP}}
	\newcommand{\sharpp}{\#\mathcal{P}}
	\newcommand{\bin}[1]{\left\llbracket#1\right\rrbracket}
	\newcommand{\bind}[1]{\bin{#1}_d}
	\newcommand{\str}{\mathbf}
	\DeclareMathOperator{\inte}{int}
	\newcommand{\abs}[1]{\left|#1\right|}
	\newcommand{\norm}[1]{\left\|#1\right\|}
	\newcommand{\bnorm}[1]{\big\|#1\big\|}
	
	\DeclareMathOperator{\diam}{diam}

\begin{document}
	\title[Complexity of integrable functions]{Complexity theory for spaces\\ of integrable functions}

	\author[F.~Steinberg]{Florian Steinberg}	
	\address{Schloßgartenstraße 7, 64289 Darmstadt, Germany}	
	\email{steinberg@mathematik.tu-darmstadt.de}  
	\thanks{The author was supported by the IRTG 1529 and EU IRSES 294962 COMPUTAL}	


	\subjclass{
		• \textbf{Mathematics of computing$\sim$Integral calculus}   • Mathematics of computing$\sim$Differential calculus   • Theory of computation$\sim$Models of computation
 	}


	\begin{abstract}
		\noindent This paper investigates second-order representations in the sense of Kawamura and Cook for spaces of integrable functions that regularly show up in analysis.
		It builds upon prior work about the space of continuous functions on the unit interval:
		Kawamura and Cook introduced a representation inducing the right complexity classes and proved that it is the weakest second-order representation such that evaluation is polynomial-time computable.

		The first part of this paper provides a similar representation for the space of integrable functions on a bounded subset of Euclidean space:
		The weakest representation rendering integration over boxes is polynomial-time computable.
		In contrast to the representation of continuous functions, however, this representation turns out to be discontinuous with respect to both the norm and the weak topology.

		The second part modifies the representation to be continuous and generalizes it to {\itshape Lp}-spaces.
		The arising representations are proven to be computably equivalent to the standard representations of these spaces as metric spaces and to still render integration polynomial-time computable.
		The family is extended to cover Sobolev spaces on the unit interval, where less basic operations like differentiation and some Sobolev embeddings are shown to be polynomial-time computable.

		Finally as a further justification quantitative versions of the \aaT and \fkT Theorems are presented and used to argue that these representations fulfill a minimality condition.
		To provide tight bounds for the \fkT Theorem, a form of exponential time computability of the norm of $\Lp$ is proven.
	\end{abstract}
	\maketitle


\section{Introduction}
	
	Classical computability and complexity theory are indispensable tools of theoretical computer science with numerous applications throughout computer science and discrete mathematics.

	In many cases, however, it is desirable to also be able to consider computations over continuous structures.
	Engineers want to use computers to solve partial differential equations that describe their problems.
	The take of mathematics on this is Numerics.
	Usually implementations of algorithms from numerics are done using floating point arithmetics.
	From the point of view of a logician this leads to a gap between between the mathematical part and the implementation:
	Proofs of convergence of an algorithm in numerics regularly rely on properties of real numbers that are not reflected by the floating point numbers.
	This can lead to uncontrollable error propagation in some situations (see the introduction of \cite{SchroederPhD} for examples).
	In other words: Floating point arithmetics are not the appropriate model of computation for Numerics.
	Furthermore, the finiteness of the set of machine numbers eludes a mathematically rigorous and comprehensible description.
	This inhibits the existence of proofs of incomputability and realistic and rigorous notions of resource consumption for computations using floating point arithmetics.
	
	These problems are well aware to numerical scientists and have been addressed by the development of interval arithmetics, multiple precision arithmetics, etc. which in turn are only partial solutions of the problem from our point of view.
	Computable analysis provides a model that fulfills our requirements by replacing a finite description of an object by a function that delivers on demand information about the object.
	Computations on objects can then be modeled by programs that are allowed subroutine calls to some function describing the object.
	Computable Analysis reaches back to \cite{turing1936computable,MR0072080,MR0086756}, the established mathematically rigorous description are Weihrauch's representations and his type two theory of effectivity (TTE) \cite{MR1795407}.
	Many of the results from computable analysis meet the intuition of numerical scientists.
	For example the empirical experience that testing for equality is not a good idea and should be replaced by an epsilon test is reflected in undecidability of equality in computable analysis but decidability of the fuzzy, multivalued version.

	A special case of complexity theory on continuous structures, namely a complexity theory for real functions, was developed by Ker-I Ko and Harvey Friedman \cite{MR666209}.
	With some technical efforts, their approach can be formulated within the TTE framework \cite{MR1795407}.
	However, complexity wise computations on real functions via the TTE are problematic.
	Thus, operators on real functions were investigated in a point wise manner instead.
	Important results were achieved:
	Ko and Friedman succeeded to prove that parametric maximization of continuous functions preserves polynomial-time computability point wise if and only if $\PP=\NP$ \cite{MR666209}.
	Whether or not the latter is the case is one of the millenium problems.
	Friedman related integration to the even stronger $\FP$ vs. $\sharpp$ problem \cite{MR748898}.
	These results remain true if the operators are restricted to smooth input, which deviates from the expectations of scientist from applied fields:
	Integration is regarded feasible, at least if smoothness assumptions are imposed, while maximization is considered more difficult.

	Of course it is also desirable to do uniform complexity theory on function spaces.
	Unfortunately the TTE turns out to yield a too restrictive model of fast computations on many structures.
	This can only partially be fixed without changing the setting \cite{MR1952428,MR2090390}.
	The problem is that the TTE only allows sequential access to the input.
	For instance for fast evaluation of a real function random access is necessary.
	An appropriate framework for complexity theory on more general spaces was more recently introduced by Kawamura and Cook \cite{MR2743298}.
	It relies on second-order complexity theory, that is complexity theory for functionals on Baire space (cf. \cite{MR0411947,MR1826285}), while maintaining the general idea of computable analysis to encode objects by integer functions.
	For this reason these encodings are called second-order representations.
	The framework is well accepted and investigated \cite{MR3239272,MR3219039,MR3259646}.
	Kawamura and Cook also specified a canonical second-order representation of the space of continuous functions on the unit interval, proving it to bring forth the same complexity classes and to be the minimal second-order representation to have polynomial-time evaluation up to polynomial-time reductions.

	While continuous functions are an important starting point, more general functions are needed for many applications.
	The emphasis on evaluation seems misplaced:
	It requires the functions considered to be continuous and polynomial-time evaluation does not suffice to carry out other important operations, like integration, effectively.
	Furthermore, the sets of functions should be considered as spaces:
	A representation induces a topology and this topology should fit the natural topology on the function space.
	If the representation is continuous and open, computability of an operator is a refinement of continuity of this operator.
	In this case there is hope that results from numerical analysis can be lifted to the computability level.
	This paper encounters discontinuous representations and discards them for this reason.

	Actually openness is sufficient but not necessary for the above hold.
	The appropriate notion from computable analysis would be the notion of admissibility of a representation.
	In the cases that turn up in this paper, admissibility is the same as equivalence to a Cauchy representation.
	To be able to hope for results from analysis to lead to algorithms that use bounded resources, a complexity theoretical equivalent of admissibility would be needed.
	Such a notion, however, does not yet exists.
	This paper very briefly mentions a representation that is the antithesis of having such a property at the end of Section \ref{sec:the singular representation}.
	The representations considered in the later chapters, in contrast, are probable candidates for having such a property.
	
	\subsection*{Content and organization of this paper}
		This paper specifies several second-order representations of spaces of integrable functions that appear in practice.
		All these representations provide oracle access to approximations of the integrals of over dyadic intervals or boxes.
		They differ, however, in the length a name of a function is given.
		On the one hand, this can be understood to modify the density of information and the time allotted for a computation on a function,
		on the other hand the length provides additional information and can be understood as \lq enrichment of data\rq \cite{MR717246}.

		Before we talk about the structure of the paper let us informally describe the representation for $\Lp$ in some more detail.
		The standard representation of the continuous functions on the unit interval establishes the following model of computation:
		A \lq program\rq\ computing a continuous function $f$ takes a rational number $r$ and a rational precision requirement $\varepsilon$ and returns an $\varepsilon$-approximation to $f(r)$ as well as a rational number $\delta$ such that all approximations stay valid whenever the input changed by less than $\delta$ (delta may only depend on $\varepsilon$, not on $r$).
		This paper claims that for functions from $\Lp$ the following is the right model of computation:
		A \lq program\rq\ computes an $\Lp$-function $f$ if it takes a rational box $[r,s]$ and a rational precision requirement $\varepsilon$ and returns an $\varepsilon$-approximation to $\int_r^sf\dd\lambda$ as well as a rational $\delta$ such that whenever the function is shifted by less than $\delta$ in the argument, it does not change more than $\varepsilon$ in $\Lp$-norm.
		The more straight forward version that $\delta$ is such that whenever the box only changes a little, the integrals do not change to much is also considered but disregarded as a discontinuous representation.

		In the later chapters, the focus shifts to justification and general recipes for how to construct useful representations.
		Interestingly, classification results for compact sets are of importance for these constructions.
		Quantitative versions of such results for concrete spaces can be connected to optimal running times of the metric with respect to any second-order representation.
		The kind of classification results that turn up have been investigated from different points of views:
		Approximation theory asked similar questions and comparable theorems turned up when constructive mathematicians tried to make analysis constructive \cite{MR0112032,MR1262128,MR804042}.

		This paper is structured as follows:
		The remainder of the first section lists some of the facts from computable analysis and real complexity theory that are regularly used throughout the paper.
		In particular it introduces Cauchy representations and the standard representations of the continuous functions on the unit interval and recollects some of the properties.
		This is mostly for easy reference and to fix a notation where more than one is common.
		
		\Cref{sec:the singular representation} introduces the singular representation:
		The weakest representation of the integrable functions that allows for the computation of integrals over boxes in polynomial-time.
		First in one dimension (Definitions~\ref{def:singularity modulus} and \ref{def:singular representation}), and then in full generality for arbitrary dimensions (Definition~\ref{def:singular representation d}).
		Theorem~\ref{resu:minimality of the singular representation d} proves that this representation is indeed minimal with respect to polynomial-time reduction.
		The singular representation is proven to be discontinuous in Theorem~\ref{resu:discontinuity norm d}.

		In \Cref{sec:a second-order representation of lp} a family of representations of the spaces $\Lp(\Omega)$ is defined and investigated.
		First, the $\Lp$-modulus, a replacement for the modulus of continuity, is discussed (Defintion~\ref{def:lp-moduli}).
		Then a representation of $\Lp(\Omega)$ is defined (Definition~\ref{d:xip}) and Theorem~\ref{resu:equivalence to the standard representation} proves it to be computably equivalent to the Cauchy representation.
		
		A straightforward extension to the Sobolev spaces $W^{m,p}([0,1])$ (Definition~\ref{d:ximp}) is presented and investigated in \Cref{sec:sobolev spaces}.
		The inclusions of the Sobolev spaces into the continuous and into the integrable functions are shown to be polynomial-time computable in Theorems~\ref{resu:sobolev functions as lp functions} and \ref{resu:Sobolev functions as continuous functions} for one derivative, and in Theorems~\ref{resu:higher sobolev functions as lp functions} and \ref{resu:higher sobolev functions as continuous functions} for higher derivatives.
		Corollary~\ref{resu:differentation} deduces that differentiation is polynomial-time computable.

		Section~\ref{sec:motivationg the use of the lp-modulus} explores minimality properties of the representations at hand:
		It introduces the concept of metric entropy (Definition~\ref{def:metric entropy and size}) and in Theorem~\ref{resu:metric entropy and complexity} proves a result that connects the metric entropy of a compact space to the minimal running time of the metric.
		This theorem is used as motivation to examine quantitative versions of theorems classifying the compact subsets of function spaces.
		Two results of this form are presented: A version of the \aaT Theorem~\ref{resu:aaT} which is already known from approximation theory and a version of the \fkT Theorem~\ref{resu:fkT} that, to the knowledge of the author, has not been stated in this generality before.
		To provide a tight upper bound for the latter (Theorem~\ref{resu:improved upper fkT}), a slightly modified representation is introduced (Definition~\ref{d:xipd}), and a strong form of exponential time computability of the norm on $\Lp$ with respect to this representation is proven in Theorem~\ref{resu:exponential time computability of the norm}.

	\subsubsection*{Sources and further readings}
		For the understanding of this paper a solid basic knowledge of computability and complexity theory is required.
		One of many excellent sources for read-up is \cite{MR2500087}.
		The topology needed and basics about metric spaces can be found for instance in \cite{munkres2000topology}.
		For the understanding of some results measure theory is necessary (for instance \cite{MR1787146}) and for the latter chapters it is beneficial to know basics about $\Lp$- and Sobolev-spaces \cite{MR2759829}.
		Furthermore, basics of real computability theory, in particular Weihrauch's type two theory of effectivity (TTE), are beneficial for understanding.
		All that is needed and more is described in detail in \cite{MR1795407}.
		For additional material on second-order complexity theory see for example \cite{MR0411947} and \cite{MR1374053}.
		For further information about the framework of Kawamura and Cook and how to apply this to computable analysis see \cite{MR2743298}.

		The results presented here are from the authors PhD-Thesis \cite{SteinbergPhD} and some were already mentioned in \cite{CIE2016}.

	\subsubsection*{Basic notational conventions}

		Fix the finite alphabet $\albe:=\binalbe$.
		The following subsets of the set $\albe^*$ of finite strings of zeros and ones are of relevance:
		\begin{description}
			\item[$\mathbf{\NN=\{\sdone,\sdone\sdzero,\sdone\sdone,\ldots\}}$]the set of strictly positive \demph{integers in binary} representation.
			\item[$\mathbf{\omega:= \{\varepsilon,\sdone,\sdone\sdone,\ldots\}}$]the set of positive \demph{integers in unary}, where $\varepsilon$ denotes the empty string interpreted as zero.
		\end{description}
		We denote elements of $\albe^*$ by $\str a,\str b, \ldots$ and elements of $\NN$ and $\omega$ by $n,m,\ldots$.
		If this leads to ambiguity we use $\sdone^n$ with $n\in\NN$ for the elements of $\omega$.
		Let $\length\cdot: \albe^*\to \omega$ denote the \demph{length function} replacing all $\sdzero$s by $\sdone$.
		To compute on $\NN$ and $\omega$ we use the following \demph{encodings} (i.e. \demph{notations} in the sense of Weihrauch \cite{MR1795407}):
		For $\NN$ the function $\nu_{\NN}:\Sigma^*\to \NN$ that eliminates leading zeros.
		For $\omega$ the function $\nu_{\omega}(\str a) := \length{\nu_{\NN}(\str a)}$.
	
		Computations on products are handled via pairing functions.
		Fix some \demph{pairing function} $\langle\cdot,\cdot\rangle:\albe^*\times \albe^*\to \albe^*$ (that is: Some bijective, polynomial-time computable function with polynomial-time computable projections).
		Furthermore, the pairing function is required to be monotone in both arguments, i.e. whenever the length of one of the input strings is increased, the length of the output string will not decrease.
		The standard pairing functions fulfill all of these requirements.
		The corresponding \demph{pairing of string functions} is defined as follows:
		\[ \langle \varphi,\psi\rangle(\str a) := \langle\varphi(\str a),\psi(\str a)\rangle. \]
		This function is bijective.

		The set of \demph{real numbers} is denoted by $\RR$.
		For $x\in \RR$ let $\lfloor x \rfloor$ resp.\ $\lceil x \rceil$ denote the largest integer smaller or equal resp.\ the least integer larger or equal to $x$.
		The binary logarithm of a number $x$ is denoted by $\lb(x)$.
		The following subsets of the real numbers are of importance to this work:
		\begin{description}
			\item[$\ZZ$] the set of \demph{integers}.
			\item[$\DD$] the set of numbers that can be written as $\frac r{2^n}$ with $r,n\in\NN$ called  \demph{dyadic numbers}.
		\end{description}
		These sets are countable and can be handled by discrete computability and complexity theory via encodings.
		For the set $\ZZ$ use the encoding $\nu_{\ZZ}(\sdone\str a) := \nu_{\NN}(\str a)$ resp. $\nu_{\ZZ}(\sdzero\str a):=-\nu_{\NN}(\str a)$.
		For $\DD$ use the encoding $\bin{\str c} :=\frac{\nu_{\ZZ}(\str a)}{2^n}$ if $\str c= a_1\sdone a_2\sdone\ldots \sdone a_m\sdzero a_{m+1}\sdone \ldots \sdone a_{n+m}\sdzero\ldots\sdzero$, and $a_2=\sdone$, i.e. the binary expansion with comma.
		This encoding is chosen such that it allows arbitrary long codes while approximations to the number can always be read from a short beginning segment.

		For any dimension $d\in\NN$ define an encoding of $\DD^d$ by $\bind{\langle\str a_1,\langle \str a_2, \ldots,\langle\str a_{d-1},\str a_d\rangle\ldots\rangle\rangle} := (\bin{\str a_1}, \ldots, \bin{\str a_d})$. Since the dimension $d$ is usually fixed, it is often omitted.
		\demph{Dyadic boxes}, i.e. boxes with dyadic vertice coordinates and edges parallel to the axes, are denoted as
		\[ [\str a,\str b] := \left[\bind{\str a}, \bind{\str b}\right] = \{x\in \RR^d\mid \bind{\str a}\leq x\leq \bind{\str b}\}, \]
		where the inequalities have to be understood component wise.

		For some $\Omega\subseteq \RR^d$ denote the set of continuous functions from $\Omega$ to $\RR$ by $\C(\Omega)$.
	\subsection{Representations}

		Encodings allow computations on countable structures using discrete computability theory.
		Many of the spaces one would like to compute over, however, are uncountable.
		For instance the real numbers, or, to mention a compact one, the unit interval.
		Computable analysis overcomes this difficulty by encoding elements by infinite objects (infinite binary strings or string functions) instead of strings \cite{MR1795407}.
		The \demph{Baire space} is the space of all string functions $(\albe^*)^{\albe^*}$ equipped with the product topology and denoted by $\B$.

		\begin{defi}
			A \demph{representation} of a space $X$ is a partial surjective mapping $\xi:\subseteq\B\to X$.
			The elements of $\xi^{-1}(x)$ are called the \demph{names} of $x$.
		\end{defi}

		A space with a fixed representation is called a represented space.
		Like for topological spaces the representation is only mentioned explicitly if necessary to avoid ambiguities.
		An element of a represented space is called \demph{computable} if it has a computable name.
		It is said to lie within a complexity class if it has a name from that complexity class.

		On one hand, any represented space carries a natural topology: The final topology of the representation.
		On the other hand, one often looks for a representation suitable for a topological space.
		It is reasonable to require such a representation to induce the topology the space is equipped with.
		For this, \demph{continuity} is necessary but not sufficient.
		Continuity together with \demph{openness} is sufficient but not necessary.
		A related concept from computable analysis is \demph{admissibility} which for all representations this paper is concerned with is the same as continuous equivalence to Cauchy representations (to be introduced below) \cite{MR1795407}.
		It implies continuity but not openness (see \cite{SchroederPhD,MR1923914} for admissibility and \cite{MR1923900} for its connection to openness).

		Recall from the introduction that $\NN\subseteq \Sigma^*$.
		\begin{defi}\label{def:metric spaces}
			Let $\M:=(M,d,(x_m)_{m\in\NN})$ be a triple such that $(M,d)$ is a complete separable metric space and $(x_m)_{m\in\NN}$ is a dense sequence.
			Define the \demph{Cauchy representation} $\xi_{\M}$ of $M$: A string function $\varphi\in\B$ is a $\xi_{\M}$-name of $x\in M$ if and only if
			\[ \forall n\in \NN : d(x,x_{\varphi(n)}) < 2^{-n}. \]
		\end{defi}
		Cauchy representations are continuous and open with respect to the metric topology.

		Recall the pairing function $\langle\cdot,\cdot\rangle:\B\times\B\to\B$ on string functions from the introduction.
		\begin{defi}\label{def:product representation}
			Let $\xi_X$ and $\xi_Y$ be representations of spaces $X$ and $Y$.
			Define the \demph{product representation} $\xi_{X\times Y}$ of the product $X\times Y$ by
			\[ \xi_{X\times Y}(\langle\varphi,\psi\rangle) := (\xi_X(\varphi),\xi_Y(\psi)) \text{, whenever $\varphi\in\dom(\xi_X)$ and $\psi\in\dom(\xi_Y)$}. \]
		\end{defi}
		This construction is used self-evidently throughout the paper.
		
		\subsection{Second-order complexity theory}\label{sec:sub:second-order complexity theory}

		Computing functions between represented spaces is done by operating on names and computing functions on Baire space:
		\begin{defi}\label{def:realizer}
			Let $\xi_X$ and $\xi_Y$ representations of spaces $X$ and $Y$.
			A partial function $F:\subseteq\B\to\B$ is called a \demph{realizer} of a function $f:X\to Y$, if
			\[ \varphi\in\xi_X^{-1}(x) \quad \Rightarrow \quad F(\varphi)\in \xi_Y^{-1}(f(x)). \]
		\end{defi}

		That is: A realizer translates names of $x$ into names of $f(x)$.
		$F$ being a realizer of a function $f$ can be visualized by the diagram in \Cref{fig:diagram}.
		However, the domain of $F$ is allowed to be bigger than that of $\xi_X$.
		Therefore, $F$ being a realizer of $f$ does not translate to the diagram being commutative in the usual way.

		\begin{wrapfigure}{r}{2.5cm}
			\xymatrix{
				\B \ar[r]^{F}\ar[d]_{\xi_X} & \B\ar[d]^{\xi_Y} \\
				X \ar[r]_f & Y,
			}
			\vspace{-.2cm}
			\caption{}\label{fig:diagram}
		\end{wrapfigure}

		On the Baire space there exists a well-established computability theory originating from \cite{MR0051790}, see \cite{MR2143877} for an overview.
		A functional $F:\subseteq\B\to\B$ is called computable if there is an oracle Turing machine $M^?$ such that $M^\varphi(\str a) = F(\varphi)(\str a)$ for all string functions $\varphi$ from the domain of $F$.
		Or spelled out: The computation of $M^?$ with oracle $\varphi$ and on input $\str a$ halts with the string $F(\varphi)(\str a)$ written on the output tape.
		A function between represented spaces is called \demph{computable} if it has a computable realizer.

		Complexity theory for functionals is called \demph{second-order complexity theory}.
		It was originally introduced by Mehlhorn \cite{MR0411947}.
		This paper uses a characterization via resource bounded oracle Turing machines due to Kapron and Cook \cite{MR1374053} as definition.
		The convention for time consumption of oracle queries is the following: When a query is asked, the answer is written on the answer tape within one time step, only reading it requires further time.
		Such a machine is granted time depending on the size of the input.
		The string functions are considered the input.

		\begin{defi}
			The \demph{size} or \demph{length} $\flength{\varphi}:\omega\to\omega$ of a string function $\varphi\in\B$ is defined by
			\[ \flength{\varphi}(\sdone^n) := \max\{\length{\varphi(m)}\mid \length m \leq \sdone^n\}. \]
		\end{defi}
		For instance: Each polynomial-time computable string function is of polynomial size.

		A running time is a mapping that assigns to sizes of the inputs an allowed number of steps.
		Therefore, it is of type $\omega^\omega\times \omega \to \omega$.
		The subclass of running times that are considered polynomial, namely second-order polynomials, are recursively defined as follows:
		\begin{itemize}
			\item For $p$ a positive integer polynomial $(l,n)\mapsto p(n)$ is a second-order polynomial.
			\item If $P$ and $Q$ are second-order polynomials, so are $P+Q$ and $P\cdot Q$.
			\item If $P$ is a second-order polynomial, then so is $(l,n)\mapsto l(P(l,n))$.
		\end{itemize}
		An example for a second-order polynomial is the mapping $(l,n)\mapsto l(l(n^2+5)+l(l(n)^2))$.
		Second order polynomials have turned up independently from second-order complexity theory (compare for instance \cite{MR1462200}).

		\begin{defi}\label{def:polynomial-time computable}
			A functional $F:\subseteq\B\to\B$ is \demph{polynomial-time computable}, if there is an oracle Turing machine $M^?$ and a second-order polynomial $P$ such that for all string functions $\varphi\in\dom(F)$ and strings $\str a$ the computation of $M^{\varphi}(\str a)$ terminates within at most $P(\flength\varphi,\length{\str a})$ steps.
		\end{defi}
		A function $f:X\to Y$ between represented spaces $X$ and $Y$ is called \demph{polynomial-time computable} if it has a polynomial-time computable realizer (compare Definition~\ref{def:realizer}).

		An important special case is the following:
		If $\xi$ and $\xi'$ are representations of  the same space $X$, then $\xi$ is called \demph{polynomial-time reducible} to $\xi'$ if the identity from $(X,\xi')$ to $(X,\xi)$ is polynomial-time computable.
		A realizer of the identity is called a \demph{translation} and we often say \demph{$\xi'$ is translatable to $\xi$} to express that $\xi$ is reducible to $\xi'$, as this avoids confusion if the directions are important.
		The representations are \demph{polynomial-time equivalent} if polynomial-time computable translations in both directions exist.
		If the translations are merely computable resp.\ continuous, one speaks of computable resp.\ continuous reduction and equivalence.

	\subsection{Second-order representations}

		The length $\flength{\varphi}(\sdone^n) =\max\{\length{\varphi(m)}\mid \length m \leq \sdone^n\}$ of a string function cannot be computed from the string function in polynomial-time:
		To find the maximum in the definition $\varphi$ has to be queried an exponential number of times.
		For many applications, polynomial-time computability of the length of names is desirable.
		\begin{defi}\label{def:length monotone}
			A string function $\varphi$ is called \demph{length-monotone} if for all strings $\str a$ and $\str b$
			\[ \length{\str a} \leq \length{\str b} \quad\Rightarrow\quad \length{\varphi(\str a)} \leq \length{\varphi(\str b)}.\]
			The set of length-monotone string functions is denoted by $\Sigma^{**}$.
		\end{defi}
		For a length-monotone string function it holds that $\flength{\varphi}(\length{\str a}) = \length{\varphi(\str a)}$,
		thus the length function restricted to $\Sigma^{**}$ is polynomial-time computable.
		\begin{defi}\label{def:second-order representation}
			A \demph{second-order representation} is a representation whose domain is contained in $\Sigma^{**}$.
		\end{defi}
		Equivalently: A second-order representation $\xi$ of a space $X$ is a partial surjective mapping $\xi:\Sigma^{**}\to X$ from the length-monotone string functions to the space.
		The prefix \lq second-order\rq\ is for applicability of second-order complexity theory, and does not indicate the use of higher order objects than for regular representations.

		The restriction of a representation to the length-monotone functions is usually still surjective and thus a second-order representation.
		All representations this paper is concerned with are second-order representations.
		For brevity \lq second-order\rq\ is sometimes omitted.

		Recall that we fixed an encoding $\bin{\cdot}$ of the dyadic numbers in the introduction.
		\begin{exa}[The standard representation of reals]\label{ex:standard representation of the reals}
			Define a second-order representation $\xi_{\RR}$ of $\RR$ by letting a length-monotone string function $\varphi$ be a name of $x$ if for all $n\in\NN$
			\[ \abs{\bin{\varphi(\sdone^n)} - x} < 2^{-n}. \]
		\end{exa}
		A proof that this second-order representation induces the established notions of computability and polynomial-time computability of reals and real functions (i.e. the notions from \cite{MR1137517} or \cite{MR1795407}) can be found in \cite{MR2275414}.
		It is computably equivalent to the Cauchy representation of the reals from Definition~\ref{def:metric spaces} if the standard enumeration of the dyadic numbers is chosen as dense sequence.
		Polynomial-time equivalence fails since the input is encoded in unary, not in binary.

		The pairing functions are carefully chosen such that the second-order representations are closed under the products from Definition~\ref{def:product representation}.
		All the encodings from the introduction assign arbitrary big codes to each element.
		Since this paper only considers representations whose names return codes from one of these, an arbitrary name can always be padded to a length monotone one.
		Therefore, all representations this paper introduces restrict to $\Sigma^{**}$ and are introduced as second-order representations right away.
		
	\subsection[The standard representation of continuous functions]{The standard representation of \texorpdfstring{$C([0,1])$}{continuous functions}}

		Denote the supremum norm on $\RR^d$ by $\abs\cdot_\infty$ and fix some bounded $\Omega\subseteq\RR^d$.
		\begin{defi}\label{def:modulus of continuity}
			A function $\mu:\omega\to\omega$ is called a \demph{modulus of continuity} of $f\in C(\Omega)$ if for all $x,y\in\Omega$ and $n\in\omega$
			\[ \abs{x-y}_\infty\leq 2^{-\mu(n)} \quad\Rightarrow\quad |f(x)-f(y)| < 2^{-n} \]
			and  $\mu(n)\neq 0 \Rightarrow \mu(n+1)>\mu(n)$, that is: $\mu$ is strictly increasing whenever non-zero.
		\end{defi}
		This modulus should be called modulus of uniform continuity to distinguish it from a point-wise modulus of continuity.
		However, point-wise moduli are not mentioned in this work and we omit the \lq uniform\rq\ for brevity.

		Any continuous function on a compact set has a modulus of continuity.
		On a connected set a function is Hölder continuous if and only if it has a linear modulus of continuity, and Lipschitz continuous if and only if it has a modulus of the form $\mu(n) = n+C$.
		If the set is convex, $2^C$ is a Lipschitz constant.

		\begin{rem}
			The definition slightly differs from the most common one (compare \cite{MR1137517} or \cite{MR1795407}):
			Usually, there is no growth condition on the modulus.
			However, whenever $\mu$ is a modulus of continuity, $\Omega$ is convex and $\mu(n)-1$ is not negative, it is a valid value for the modulus of continuity on $n-1$.
			Thus, the condition of being strictly increasing when non-zero is a reasonable one and in particular fulfilled by the least modulus of continuity.
			Its importance becomes apparent in the proof of Theorem~\ref{resu:fkT}.
		\end{rem}

		Recall that $\DD$ denotes the set of dyadic numbers and that computations on $\DD$ are carried out via the encoding $\bin\cdot$ fixed in the introduction.
		The following result is the starting point of many generalizations:		
		\begin{thm}[\cite{MR1137517}]\label{resu:characterization of polynomial-time computable functions}
			A function $f:[0,1] \to \RR$ is polynomial-time computable if and only if both of the following are fulfilled:
			\begin{itemize}
				\item There is a polynomial-time computable function $\varphi:\DD \times \omega\to \DD$ such that for any $r\in[0,1]\cap \DD$
				\[ \abs{\varphi(r,\sdone^n) -f(r)} < 2^{-n}. \]
				\item The function allows a polynomial modulus of continuity.
			\end{itemize}
		\end{thm}
		
		This theorem can be used to define complexity of functions between arbitrary effective metric spaces \cite{MR1795248}.
		Another application is to show that the following definition leads to the usual set of polynomial-time computable functions on the unit interval.
		Recall that the length $\flength{\varphi}$ of a length-monotone string function is given by $\flength{\varphi}(\length{\str a}) = \length{\varphi(\str a)}$.
		\begin{defi}\label{def:standard rep}
			Define the \demph{standard representation $\xic$ of $C([0,1])$}:
			A string function $\varphi\in\Sigma^{**}$ is a $\xic$-name of $f$ if for all strings $\str a$ with $\bin{\str a}\in[0,1]$ and all $n\in\NN$
			\[ |\bin{\varphi(\langle \str a,\sdone^n\rangle)}-f(\bin{\str a})| < 2^{-n} \]
			 and $\flength{\varphi}$ is a modulus of continuity of $f$.
		\end{defi}

		$\C([0,1])$ is a metric space and the Cauchy representation with respect to the standard enumeration of the rational polynomials as dense sequence (cf. Definition~\ref{def:metric spaces}) induces the metric topology.
		The following is closely connected to the well-known computable Weierstraß approximation theorem (compare for instance \cite{MR1137517}):
		\begin{thm}
			$\xic$ is computably equivalent to the Cauchy representation of $\C([0,1])$.
		\end{thm}
		In particular, $\xic$ is a continuous mapping.
		The strict inequality in the definition of the modulus of continuity guarantees that $\xic$ is an open mapping.
		The standard representation has been characterized as the weakest representation that permits polynomial-time evaluation up to polynomial-time equivalence.
		Recall the evaluation operator given by
		\[ \eval: C([0,1]) \times [0,1] \to \RR, \quad (f,x)\mapsto f(x). \]
		Here and in the following, both $[0,1]$ and $\RR$ are equipped with the standard representation $\xi_{\RR}$ of the real numbers from Example~\ref{ex:standard representation of the reals} and its range-restriction.
		For a subset $F\subseteq C([0,1])$ the evaluation operator on $F$ is the restriction of the above operator to $F\times [0,1]$.

	 	\begin{thm}[minimality]\label{resu:minimality of the standard representation}
	 		For a second-order representation $\xi$ of a subset $F\subseteq C([0,1])$ the following are equivalent:
	 		\begin{itemize}
	 			\item $\xi$ renders the evaluation operator on $F$ polynomial-time computable.
	 			\item $\xi$ is polynomial-time translatable to the range-restriction of $\xic$ to $F$.
	 		\end{itemize}
	 	\end{thm}
	 	A proof of this theorem for $F=C([0,1])$ can be found in \cite{MR2743298} and is easily seen to work for an arbitrary $F$ as well.
	 	Note that the computable version of this theorem holds in a more general setting: $[0,1]$ and $\RR$ can be replaced by arbitrary represented spaces $X$ and $Y$ and $\xic$ by the function space representation of the continuously representable functions from $X$ to $Y$.
	 	Under suitable assumptions about the represented
                spaces, this is a well behaved representation of the
                continuous functions from $X$ to $Y$.	


\section{The singular representation}\label{sec:the singular representation}
	
	Fix some bounded measurable set $\Omega\subseteq \RR^d$.
	Recall that $\LLL1{\Omega}$ denotes the set of functions on $\Omega$ integrable with respect to the Lebesgue measure $\lambda$, where functions are identified if they coincide almost everywhere.
	Equipped with the norm $\norm{f}_1:= \int_\Omega\abs{f}\dd\lambda$ the space $\LLL1{\Omega}$ is a Banach space.
	This section specifies the weakest representation of $\LLL1{\Omega}$ that renders integration polynomial-time computable.
	More formally the following operator is supposed to be polynomial-time computable:
	\begin{equation}\label{the integration operator}\tag{INT}
		\inte:\LLL1{\Omega}\times \Omega^2 \to \RR,\quad (f,x,y)\mapsto \int_{[x,y]\cap\Omega} f\dd\lambda,
	\end{equation}
	where $[x,y]$ denotes the smallest box with edges parallel to the axis and corners $x$ and $y$.
	Here, $\RR^d$ is equipped with the $d$-fold product of the standard representation of the real numbers and $\Omega$ with its range-restriction.
	
	First consider the case $\Omega = [0,1]$:
	Define an operator $\Phi:\LLL1{[0,1]} \to C([0,1])$ by
	\[ \Phi(f)(x) := \int_0^x f(t)\dd t. \]
	This defines a linear continuous operator between Banach spaces with $\|\Phi\| = 1$.
	The operator $\Phi$ translates the integration operator into the evaluation operator:
	\[ \eval(\Phi(f),x) = \inte(f,0,x), \quad \inte(f,x,y) = \eval(\Phi(f),y)-\eval(\Phi(f),x). \]
	The image of $\Phi$ is the set $\mathcal{AC}_0([0,1])$ of absolutely continuous functions that vanish in zero.
	Furthermore, $\Phi$ is injective and therefore invertible on its image.

	From the above it follows that $\Phi^{-1}\circ\xic|^{\mathcal{AC}_0([0,1])}$ is a minimal representation:
	Whenever $\xi$ renders integration polynomial-time computable, $\Phi\circ \xi$ renders evaluation polynomial-time computable.
	Thus, the polynomial-time translation from $\xic|^{\mathcal{AC}_0([0,1])}$ to $\Phi\circ \xi$ that exists by the minimality of $\xic$ from Theorem~\ref{resu:minimality of the standard representation} is also a polynomial-time translation from $\Phi^{-1}\circ \xic|^{\mathcal{AC}([0,1])}$ to $\Phi^{-1}\circ\Phi\circ \xi =\xi$.
	
	Since $\Phi^{-1}$ is a linear discontinuous operator between Banach spaces, this representation cannot be continuous:
	An abstract argument for this can be found in \cite{SchroederPhD}.
	This chapter specifies an alternative description of the above representation that allows for generalizations and proves that the representation and its multidimensional generalizations are discontinuous.

	\subsection{Singularity moduli}\label{sec:singularity moduli}
		With the notation from the introduction of this section:
		For $f\in\Lone(\Omega)$ a function $\mu$ is a modulus of continuity of $\Phi(f)$ if and only if
		\[ \abs{x-y}\leq 2^{-\mu(n)} \quad \Rightarrow\quad \abs{\int_x^y f\dd\lambda} < 2^{-n} \]
		and it is strictly increasing whenever it is non-zero.
		This motivates the following definition.
		Let $\Omega$ be a measurable subset of $\RR$ (it is no longer assumed to be bounded).
		For $f\in\Lone(\Omega)$ denote by $\tilde f$ the extension of $f$ to all of the real line by zero.
		The following modulus measures how bad the singularities of a function are:
		\begin{defi}\label{def:singularity modulus}
			A function $\mu:\omega\to\omega$ is called a \demph{singularity modulus} of $f\in \LLL1\Omega$, if for any $n\in\omega$ and $x,y\in \RR$
			\[ \abs{x-y}\leq 2^{-\mu(n)}\quad \Rightarrow \quad\abs{\int_{[x,y]} \tilde f\dd\lambda} < 2^{-n}. \]
		\end{defi}
		Like any continuous function on the unit interval allows a modulus of continuity, any function from $\Lone(\Omega)$ possesses a singularity modulus.
		For $\Omega =[0,1]$ any modulus of continuity of the function $\Phi(f)$ from the introduction of this section may be chosen.
		It is possible to prove that any integrable function has a singularity modulus.
		
		If $\Omega$ is bounded, the existence of a singularity modulus implies integrability.
		If the interior of $\Omega$ is unbounded the situation is more involved:
		On the one hand, there are non-integrable, but locally integrable functions that permit a singularity modulus.
		On the other hand not all locally integrable functions allow a singularity modulus.
		In the following, however, only bounded sets are considered.

		The next proposition uses $\Lp$-spaces, that are recollected in \Cref{sec:a second-order representation of lp} in more detail.
		The case $p=\infty$, is understandable if one recalls that $\LL\infty$ are the essentially bounded functions and $\norm\cdot_\infty$ the essential supremum norm.
		\begin{prop}[small moduli]\label{resu:classes with small moduli}
			For a function $f\in\LLL1{[0,1]}$ and an integer $C\in\omega$ the following hold:
			\begin{enumerate}
				\item\label{item:singmod} if $n\mapsto n+C$ is a singularity modulus of $f$, then $f\in\LLL\infty{[0,1]}$ and $\lb(\|f\|_\infty)\leq C$.
				\item\label{item:lp} If $f\in\LLL p{[0,1]}$ for some $1< p\leq \infty$ and $C> \lb(\|f\|_p)$ and $D\geq\big(1-\frac1p\big)^{-1}$ are integer constants, then $n\mapsto D(n+C)$ is a singularity modulus of $f$.
			\end{enumerate}
		\end{prop}

		\noindent For the proof recall the following theorem, a proof of which can be found in \cite{MR924157}.
		\begin{thm}[Lebesgue Differentiation Theorem]\label{resu:lebesgue differentiation thoerem}
			Let $f\in \Lone(\RR)$.
			Then for any representative $g$ of $f$ and almost all $x\in\RR$ it holds that
			\[ g(x) = \lim_{m\to\infty} 2^{m}\int_{x-2^{-m-1}}^{x+2^{-m-1}} g\dd \lambda. \]
		\end{thm}

		\begin{proof}[Proof of Propoition~\ref{resu:classes with small moduli}]
			First prove \cref{item:singmod}.
			For this assume that $n\mapsto n+C$ is a singularity modulus of $f$ and let $g$ be a representative of the function considered.
			By the Lebesgue Differentiation Theorem~\ref{resu:lebesgue differentiation thoerem} there exists a set $A\subseteq [0,1]$ of measure one such that for any $x\in A$
			\[ \abs{g(x)} = \lim_{m\to\infty} 2^{m}\abs{\int_{x-2^{-m-1}}^{x+2^{-m-1}} g\dd \lambda} = \lim_{n\to\infty} 2^{n+C}\abs{\int_{x-2^{-n+C-1}}^{x+2^{-n+C-1}} f\dd \lambda} < 2^C. \]
			This proves that $\norm f_\infty \leq 2^C$ and in particular that $f\in \LLL{\infty}{[0,1]}$. 

			To prove \cref{item:lp} use Hölder's inequality (see Corollary~\ref{cor:hoelder}) to deduce
			\[ \abs{\int_x^{x+h}f(t)\dd t} \leq \int_x^{x+h} |f(t)|\dd t \leq \|f\|_ph^{1-\frac1p}. \]
			From this it is easy to see that the assertion is true.
			It remains true for $p=\infty$ if the convention $\frac1\infty = 0$ is used.
		\end{proof}
		In particular the functions with singularity modulus of form $n+C$ for some $C$ are exactly the functions contained in $\LL\infty$.
		The class of functions with linear modulus with slope $(1-1/p)^{-1}$ contains $\Lp$, however, the inclusion is strict as can be seen by considering the function $x^{-1/p}$.
		The corresponding classes for the modulus of continuity are the Lipschitz and Hölder-continuous functions.
		
	\subsection{The singular representation in one dimension}\label{sec:sub:in one dimension}

		Recall that the information about a continuous function can be divided into two parts by the characterization of polynomial-time computable functions from Theorem~\ref{resu:characterization of polynomial-time computable functions}.
		The first part being approximations to the values on dyadic numbers and the second part being a modulus of continuity.
		Definition~\ref{def:standard rep} uses this to introduce the standard representation of continuous functions.

		The following definition carries this idea to the set of integrable functions, where integrals over dyadic intervals replace the point evaluations and the singularity modulus replaces the modulus of continuity.
		Recall that $\DD$ denotes the set numbers of the form $\frac{m}{2^n}$ for $m\in\ZZ$ and $n\in\omega$, that $\bin\cdot:\albe^* \to \DD$ is the encoding fixed in the introduction and that the length $\flength{\varphi}$ of a length-monotone string function is given by $\flength{\varphi}(\length{\str a}) = \length{\varphi(\str a)}$.
		\begin{defi}\label{def:singular representation}
			Define the \demph{singular representation} $\xi_{s}$ of $\LLL1{[0,1]}$:
			A length-monotone string function $\varphi$ is a name of $f\in\LLL1{[0,1]}$ if for all strings $\str a,\str b$ with $\bin{\str a}, \bin{\str b}\in[0,1]$
			\[ \abs{\int_{\bin{\str a}}^{\bin{\str b}} f\dd \lambda -\bin{\varphi(\langle \str a,\str b, \sdone^n\rangle)}} < 2^{-n} \]
			and $\flength{\varphi}$ is a singularity modulus of $f$.
		\end{defi}
		This definition is well posed:
		Firstly, for any distinct integrable functions there exists a dyadic interval such that their integrals over this interval differ.
		Thus, the above indeed defines a partial function.
		Secondly, any integrable function has a singularity modulus and therefore the mapping is surjective.

		It can easily be verified that this representation renders the vector space operations of $\Lone([0,1])$ polynomial-time computable.
		The representation $\xis$ is chosen such that it is polynomial-time equivalent to the representation from the introduction of this section.
		As a result, it possesses the same minimality property.
		We state this as a theorem, note however, that it is also covered by Theorem~\ref{resu:minimality of the singular representation d}, which contains a more explicit statement and a direct proof.
		\begin{thm}[minimality]
			$\xis$ is a minimal representation of $\Lone([0,1])$ such that the integration operator is polynomial-time computable.
		\end{thm}

	\subsection{Higher dimensions}

		Definition~\ref{def:singular representation} allows a straight forward generalization to higher dimensions.
		Fix some dimension $d$, let $\Omega\subseteq\RR^d$ be a bounded measurable set and recall that $\tilde f$ denotes the extension of a function to the whole space by zero.
		\begin{defi}\label{def:singularity modulus d}
			A function $\mu:\omega\to\omega$ is called a \demph{singularity modulus} of $f\in \LLL1\Omega$ if it is a singularity modulus (in the sense of Definition~\ref{def:singularity modulus}) for each of the functions
			\[ f_i(x):= \int_{\RR^{d-1}} \tilde f(x_1,\ldots x_{i-1},x,x_{i+1},\ldots, x_{d})\dd x_1\cdots \dd x_{i-1} \dd x_{i+1}\cdots \dd x_{d}. \]
			(compare \Cref{fig:sing.mod.}).
		\end{defi}
		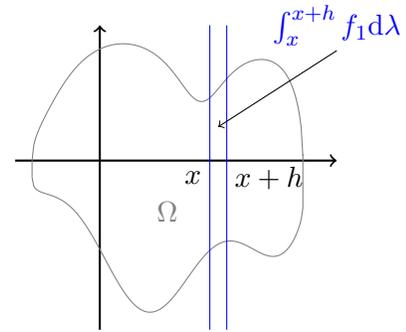
\begin{wrapfigure}{r}{5.2cm}
			\vspace{-1.2cm}
			\center
			\begin{tikzpicture}[scale=.45]
				\draw[thick,->] (-2.5,0) -- (7,0);
				\draw[thick,->] (0,-5) -- (0,4);
				\draw[domain = 0:8,color=gray] plot[samples=160] (\x-2, {sqrt(16-(\x-4)^2)*(1-1/(\x^2+1) + sin(\x) + (\x/7)^5-\x/6 +1.5- 1/((\x-5)^2+1))/2});
				\draw[domain = 0:8,color=gray] plot[samples=160] (\x-2, {1/(\x^2+1)/2 + sin(80*\x) + (\x/7)^5-\x/6 -1-sqrt(16-(\x-4)^2)/2});
				\draw[color=gray] (-2,0) -- (-2,-.5);
				\draw[color=gray] (6,.19) -- (6,-1.42);
				\node[color=gray] at (2,-1.5) {$\Omega$};
				\draw[color=blue] (3.25,-5) -- (3.25,4);
				\draw[color=blue] (3.75,-5) -- (3.75,4);
				\node at (2.75,-.5) {$x$};
				\node at (5,-.5) {$x+h$};
				\node[color=blue] at (7,4) {$\int_x^{x+h}f_1\dd\lambda$};
				\draw[->] (7,3.25) -- (3.5,1);
			\end{tikzpicture}
			\caption{$f_1$ in two dimensions.}\label{fig:sing.mod.}
			\vspace{-1.2cm}
		\end{wrapfigure}

		Recall from the introduction that for $x,y\in\RR^d$ the smallest box with corners $x$ and $y$ and edges parallel to the axis is denoted by $[x,y]$ and that the box $\left[\bind{\str a},\bind{\str b}\right]$ is abbreviated as $[\str a,\str b]$.
		\begin{defi}\label{def:singular representation d}
			Define the \demph{singular representation} $\xis$ of $\LLL1\Omega$:
			A length-monotone string function $\varphi$ is a $\xis$-name of $f \in\LLL1\Omega$ if for all strings $\str a,\str b$ with $\bin{\str a}, \bin{\str b}\in\DD^d$\\
			\begin{minipage}{.6\textwidth}
				\[ \abs{\int_{\left[\str a,\str b\right]} \tilde f\dd\lambda -\bin{\varphi(\langle \str a,\str b, \sdone^n\rangle)}} < 2^{-n} \]
			\end{minipage}\\
			and $\flength{\varphi}$ is a singularity modulus of $f$.
		\end{defi}

		Since no source for a multidimensional generalization of the minimality of the standard representation for continuous functions from Theorem~\ref{resu:minimality of the standard representation} is known to the author, a direct proof of the minimality is given for the multidimensional case.
		\begin{thm}[minimality of the singular representation]\label{resu:minimality of the singular representation d}
			For a second-order representation $\xi$ of $\LLL1\Omega$ the following are equivalent:
			\begin{itemize}
				\item $\xi$ renders the integration operator from \cref{the integration operator} polynomial-time computable.
				\item $\xi$ is polynomial-time translatable to the singular representation $\xis$.
			\end{itemize}
		\end{thm}

		\begin{proof}
			The proof is very similar to the proof of \cite[Lemma 4.9]{MR2743298}, i.e. of the minimality of $\xic$ from Theorem~\ref{resu:minimality of the standard representation}.
			First assume that $\xi$ is a representation such that the integration operator from \cref{the integration operator} polynomial-time computable.
			Describe an oracle Turing machine that whenever given a $\xi$-name $\varphi$ of a function $f\in\LLL1\Omega$ returns correct values of a $\xis$-name of $f$:
			This machine simulates a machine computing the integration operator in polynomial-time to obtain approximations to the integrals of $f$ from $\varphi$.
			
			To obtain a singularity modulus of the input function let $P$ be a second-order polynomial bounding the running time of the integration operator and $p$ a polynomial such that any $(x,y)\in\Omega^2$ has a name of length $p$ (This depends on the concrete encoding of dyadic numbers and products chosen, but exists for reasonable choices and bounded $\Omega$).
			Then $\mu:\omega\to\omega,$ $n\mapsto P(\langle|\varphi|,p\rangle,n+1)$ is a singularity modulus of $f$:		
			When queried for an approximation with quality $2^{-n-1}$ the machine computing the integration operator can at most take $\mu(n)$ steps.
			Therefore, it knows the boundaries $a$ and $b$ of the integral with precision at most $2^{-\mu(n)}$.
			Recall the definition of the singularity modulus from Definition~\ref{def:singularity modulus d}, in particular that $f_i$ was $f$ with all but the $i$-th variable integrated over.
			Since $\Omega$ is bounded, there are some $\str c$ and $\str d$ such that $\Omega\subseteq [\str c,\str d]^d$.
			Note that for any $x\in[\str c,\str d]$ and $h\in\RR$ with $\abs h\leq 2^{-\mu(n)}$ there is a dyadic vector $a = \bin{\str a}$ which is a valid $2^{-\mu(n)}$ approximation for both $x$ and $x+h$.
			
			The argument works the same for any $i$.
			Set $i=d$ from now on to simplify notation.
			Define length-monotone string functions $\varphi_a^+$ and $\varphi_a^-$ by
			\[ \varphi_a^+(\str b) := \langle\str d,\ldots,\langle \str d,\str a\rangle\ldots\rangle\text{, resp. } \varphi_a^-(\str b):=\langle\str c,\ldots,\langle \str c, \str a\rangle\ldots\rangle. \]
			Let $q$ be the approximation encoded in the output of the machine computing $\inte$ when handed $\varphi$ as function name, $\langle\varphi_a^-,\varphi_a^+\rangle$ as boundaries of the integral and $\sdone^{n+1}$ as precision requirement.

			Since $a$ is an approximation to both $x$ and $x+h$ and $[\str c, \str d]^{d-1}\times[a,a]$ is a set of Lebesgue measure zero
			\begin{align*}
				\abs{\int_x^{x+h}f_i\dd \lambda} 
				& \leq \abs{\int_{[\str c,\str d]^{d-1} \times[x,x+h]} \tilde f\dd \lambda - q} + \abs{q -\int_{[\str c,\str d]^{d-1} \times[a,a]} \tilde f\dd \lambda}
				< 2^{-n}.
			\end{align*}
			
			It is left to show that $\xis$ renders the integration operator polynomial-time computable.
			Assume a $\xis$ name $\varphi$ of a function $f$, an oracle for a box and a precision requirement $\sdone^n$ are given.
			Get approximations to the vertices of the box with precision $\sdone^{\flength\varphi(\sdone^n)+\lceil\lb(d)\rceil+1}$ and query $\varphi$ for a $2^{-n-1}$ approximation over this box.
			An easy triangle inequality argument shows that this is a valid approximation to the integral over the box.
		\end{proof}

		The result includes null sets:
		In this case $\LLL1\Omega$ only contains one element and the integration operator is the constant zero function.

	\subsection{Discontinuity}\label{sec:sub:discontinuity}

		Under reasonable assumptions, the singular representation is discontinuous.
		Since the proof of discontinuity is most naturally stated for the unit interval, we state a restricted version first:

		\begin{prop}[discontinuity]\label{resu:discontinuity norm}
			The singular representation $\xi_s$ is not continuous with respect to the norm topology on $\LLL1{[0,1]}$.
		\end{prop}

		\begin{proof}
			\quad\vspace{-.4cm}

			\noindent
			\begin{minipage}{.55\textwidth}
				\hspace{1.1cm} Consider the sequence of functions on the unit interval defined by
				\[ f_m(x):= (-1)^{\min\{k\in \NN\mid k2^{-m}\geq x\}}. \]
				That is: Divide $[0,1]$ into $2^m$ equally sized intervals and let the function values alternate between constantly being $1$ an $-1$ respectively on these intervals	(compare \Cref{fig:fm}).
				The functions $f_m$ are \quad\quad\quad
				\vspace{-.3cm}
			\end{minipage}
			\hfill
			\begin{minipage}{.45\textwidth}
				\vspace{-.4cm}
				\centering
				\begin{tikzpicture}
					\draw (-.1,1) -- (.1,1);
					\node at (-.2,1) {$1$};
					\draw (-.1,-1) -- (.1,-1);
					\node at (-.35,-1) {$-1$};
					\draw[->,thick] (0,-1.2) -- (0,1.5);
					\draw[->,thick] (0,0) -- (4.5,0);
					\node[left] at (0,0) {$0$};
					\draw (4,.1) -- (4,-.1);
					\node at (4,-.35) {$1$};
					\draw[blue] (0,-1) -- (.4,-1);
					\draw (.4,.1) -- (.4,-.1);
					\node at (.6,-.35) {$2^{-m}$};
					\draw[blue,dotted] (.4,-1) -- (.4,1);
					\draw[blue] (.4,1) -- (.8,1);
					\draw[blue,dotted] (.8,-1) -- (.8,1);
					\draw[blue] (.8,-1) -- (1.2,-1);
					\draw[blue,dotted] (1.2,-1) -- (1.2,0);				
					\node at (2,-.25) {$\ldots$};
					\draw[blue,dotted] (3.2,0) -- (3.2,-1);
					\draw[blue] (3.2,-1) -- (3.6,-1);
					\draw[blue,dotted] (3.6,-1) -- (3.6,1);
					\draw[blue] (3.6,1) -- (4,1);
					\draw[blue,dotted] (4,1) -- (4,0);
				\end{tikzpicture}
				
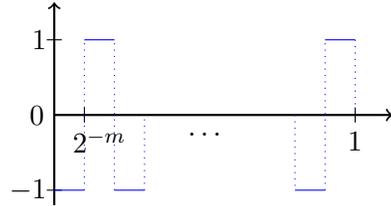
\captionof{figure}{The function $f_m$}\label{fig:fm}
				\vspace{.2cm}
			\end{minipage}
			bounded by $1$ in the norm $\norm\cdot_\infty$ and thus allow the common singularity modulus $n\mapsto n+1$ by Proposition~\ref{resu:classes with small moduli}.
			Observe that the integrals of $f_m$ over an interval is always smaller than the minimum of the length of the interval and $2^{-m}$.
			Thus, since from $\length{\langle \str a,\str b, 1^{k}\rangle}<\sdone^m$ it follows that $k<m$, it is possible to choose a name $\varphi_m$ of $f_m$ such that upon this input a string $\str c$ of length more than $\length{\langle \str a,\str b, \sdone^{k}\rangle}+1$ with $\bin{\str c}=0$ is returned.
			The sequence $\varphi_m$ converges to a name of the zero function in Baire-space.
			However, $\xis(\varphi_m) = f_m$ has norm $1$ for all $m$ and therefore does not converge to the zero function in norm.
			This proves discontinuity of $\xis$.
		\end{proof}

		\begin{thm}[discontinuity]\label{resu:discontinuity norm d}
			Whenever $\Omega\subset \RR^d$ is a bounded set with non-empty interior, the singular representation $\xis$ of $\Lone(\Omega)$ from Definition~\ref{def:singular representation} is discontinuous.
		\end{thm}

		\begin{proof}[sketch of the proof]
			Since the interior of $\Omega$ is non-empty, there exists a small box with edges parallel to the axis and dyadic endpoints completely included in $\Omega$.
			Lift the function sequence from the proof of Proposition~\ref{resu:discontinuity norm} to a box by assuming the functions to be independent of the additional variables.
			Scale this box and the functions to fit inside of $\Omega$.
			The arguments of the above proof still work for this new sequence and show discontinuity.
		\end{proof}

		This result can be strengthened to prove discontinuity of the singular representation with respect to the weak topology on $\Lone(\Omega)$.
		It is not known to the author if the final topology of $\xi_s$ coincides with any topology of $\Lone$ that has been considered before.

		Using $\xi_s$ it is possible to construct a weakest representation computably equivalent to the Cauchy representation, such that integration is polynomial time computable.
		However, this representation has very undesirable properties from a complexity theoretical point of view:
		While it renders many operations, like the metric,
                computable, only those that are already computable in
                bounded time with respect to the singular
                representation become computable in bounded time.
                

\section[Lp-spaces]{\texorpdfstring{$\Lp$}{Lp}-spaces}\label{sec:a second-order representation of lp}

		The previous chapter introduced the weakest representation $\xis$ of integrable functions such that the integration operator is polynomial-time computable and showed its discontinuity.
		However, the way the representation was introduced allows for straight forward generalizations:
		Like for the continuous functions the information was divided into a discrete part, the integrals over dyadic intervals, and a topological part: the singularity modulus.
		This section discusses a replacement for the singularity modulus which leads to a continuous representation.
		More precisely, for any $p$ a modulus is defined that exists if and only if the function is an element of $\LLL p\Omega$.

		First recall some basic facts about spaces of integrable functions: 
		Let $\lambda$ denote the Lebesgue measure of any dimension.
		In the following $\Omega$ denotes a bounded, measurable set.
		Recall that the space $\LLL p\Omega$ is the Banach space of equivalence classes of measurable functions up to equality almost everywhere such that
		\[ \norm f_p = \norm{f}_{p,\Omega} := \left(\int_\Omega\abs{f}^p\dd\lambda\right)^{\frac 1p} < \infty. \]
		And that for the case $p=\infty$ the norm $\norm\cdot_\infty$ is defined to be the essential supremum norm.
		If $\Omega$ is bounded with non-zero Lebesgue measure, then $\C(\Omega)\subsetneq \LLL p\Omega \subsetneq \LLL q\Omega$ whenever $1\leq q<p\leq\infty$.
		The inclusions are continuous.
		This can be seen using the following well known result from analysis:
		\begin{thm}[Hölder's Inequality]\label{thm:hoelder}
			For any measurable subset $\Omega\subseteq\RR^d$, any measurable functions $f$, $g$ on $\Omega$ and any $p\in[1,\infty]$ the inequality
			\[ \|fg\|_1 \leq \|f\|_p \|g\|_q \]
			holds. Where $q := \frac1{1-\frac1p}$ is the conjugate exponent of $p$ and $q =\infty$ if $p=1$.
		\end{thm}

		A corollary from this is particularly often useful for our purposes:
		\begin{cor}\label{cor:hoelder}
			For any measurable function $f$ on a measurable set $\Omega$ it holds that
			\[ \int_{\Omega} |f| d\lambda \leq \lambda(\Omega)^{1-\frac1p}\|f\|_{p}. \]
		\end{cor}
		For $\Omega =[0,1]$, $f:= \abs{g}^r$ and $p:=\frac sr$ the above proves that $\|g\|_r\leq \|g\|_s$ whenever $r<s$.

	\subsection[The Lp-modulus]{The \texorpdfstring{$\Lp$}{Lp}-modulus}
		In the following let $1\leq p <\infty$.
		Recall that $\tilde f$ denotes the extension of $f$ to all of $\RR^d$ by zero.
		For $h\in\RR^d$ the shift operator $\tau_h$ is defined by $(\tau_h f)(x) := f(x+h)$.

		\begin{defi}\label{def:lp-moduli}
			A function $\mu:\omega\to \omega$ is called an \demph{$\Lp$-modulus} of $f\in\Lp(\Omega)$ if
			\[ \abs{h}\leq 2^{-\mu(n)} \quad\Rightarrow\quad \| \tilde f - \tau_{h}\tilde f\|_{p}< 2^{-n}, \]
			and  $\mu(n)\neq 0 \Rightarrow \mu(n+1)>\mu(n)$ (i.e. it is strictly increasing whenever non-zero).
		\end{defi}
		Due to the assumption $p<\infty$ any $\Lp$-function has an $\Lp$-modulus (see for instance \cite[Lemma 4.3]{MR2759829}).

		Whenever $1\leq q \leq p<\infty$, an $\LL q$-modulus can be obtained from an $\LL p$-modulus by shifting with a constant.
		Let $f\in\C(\Omega)$ be such that its extension $\tilde f$ is continuous.
		In this case a modulus of continuity $\mu$ of $f$ is also a modulus of $\tilde f$ and can be converted into an $\Lp$-modulus of $f$:
		Whenever $|h|\leq 2^{-\mu(n)}$, then $|x-(x+h)|\leq 2^{-\mu(n)}$ and therefore
		\begin{align*}
			\|\tilde f-\tau_h\tilde f\|_p & = \left(\int_{\RR^d} |\tilde f(x)-\tilde f(x+h)|^p\dd x\right)^{\frac1p} 
			 < 2^{-n+\frac{1+\lb(\lambda(\Omega))}p}.
		\end{align*}
		Thus, $n\mapsto \mu\left(n+\lceil(1+\lb(\lambda(\Omega)))/{p}\rceil\right)$ is an $\Lp$-modulus of $f$.
		If $f$ can not be continuously extended, additional information about the function and the domain is needed to obtain an $\Lp$-modulus from a modulus of continuity (compare Lemma~\ref{resu:from cont. mod. to Lp-mod.}).

		The modulus of continuity does not contain any information about the norm of a function as it does not change under shift with a constant function.
		In contrast to that a norm-bound can be deduced from an $\Lp$-modulus.
		Recall the diameter of a set:
		\[ \diam(\Omega) :=\sup\{\abs{x-y}_\infty\mid x,y\in\Omega\}, \]
		where $\abs{\cdot}_\infty$ denotes the supremum norm on $\RR^d$.
		\begin{lem}[norm estimate]\label{resu:norm bound from modulus}
			Whenever $\mu$ is an $\Lp$-modulus of $f\in\LLL{p}{\Omega}$, then
			\[ \|f\|_p < \lceil{\diam(\Omega)}\rceil2^{\mu(0)-\frac1p}. \]
		\end{lem}

		\begin{proof}
			Fix some unit vector $e\in\RR^d$.
			The intersection of $\Omega +\lceil\diam(\Omega)\rceil e$ and $\Omega$ has zero Lebesgue measure.
			Thus:
			\begin{align*}
				2^{\frac1p}\|f\|_p & \leq \sum_{i=1}^{\lceil\diam(\Omega)\rceil2^{\mu(0)}}\|\tilde f - \tau_{2^{-\mu(0)} e}\tilde f\|_p 
				<
				\lceil\diam(\Omega)\rceil2^{\mu(0)},
			\end{align*}
			which proves the assertion.
		\end{proof}

	\subsection[Representing Lp]{Representing \texorpdfstring{$\Lp$}{Lp}}
	
		Let $\Omega\subseteq \RR^d$ be a bounded measurable set of non-zero measure.

		\begin{defi}\label{d:xip}
			Define the \demph{second-order representation $\xip$ of $\LLL p{\Omega}$}: A length-monotone string function $\varphi$ is a name of $f\in \LLL p{\Omega}$ if and only if for any strings $\str a, \str b\in\albe^*$ and $n\in \NN$
			\[ \abs{\int_{[\str a,\str b]} \tilde f\dd\lambda - \bin{\varphi(\langle \str a,\str b, \sdone^n\rangle)}} < 2^{-n}, \]
			and $\flength \varphi$ is an $\Lp$-modulus of $f$.
		\end{defi}

		Again, the vector space operations of $\Lp(\Omega)$ are easily shown to be polynomial-time computable with respect to this representation.
		Recall that for bounded $\Omega$ and $p\geq q$ it holds that $\LLL p\Omega\subseteq\LLL q\Omega$.
		Since an $\LL q$-modulus can be obtained from an $\LL p$-modulus by shifting with a constant, this inclusion is polynomial-time computable.
		The partial inverse of this inclusion is discontinuous and therefore not computable.
		To compare $\xip$ to the singular representation from Definition~\ref{def:singular representation d} note:

		\begin{prop}\label{resu: from lp-modulus to singularity modulus}
			Let $\mu$ be an $\Lp$-modulus of a function.
			Then
			\[ n\mapsto \mu(n+1+ \lb(\lambda(\Omega))+(d-1)\lb(\diam(\Omega))) \]
			is a singularity modulus of the function.
		\end{prop}

		\begin{proof}
			Let $f\in\Lp(\Omega)$ be a function and let $\mu$ be an $\Lp$-modulus of $f$.
			Recall from Definition~\ref{def:singularity modulus d} that $f_i$ denotes the function where all but the $i$-th variable has been integrated over.
			Apply the transformation rule and the version of Hölder's inequality from Corollary~\ref{cor:hoelder} to get
			\begin{align*}
				\abs{\int_x^{x+h}f_i(t)\dd t} &\leq \abs{\int_x^\infty f_i(t) \dd t - \int_{x+h}^\infty f_i(t) \dd t} \leq \int_{-\infty}^\infty\abs{f_i-\tau_{h} f_i}\dd\lambda \\
				& \leq (2\lambda(\Omega))^{1-\frac1p}\norm{f_i-\tau_{h} f_i}_p \leq (2\lambda(\Omega))^{1-\frac 1p}\diam(\Omega)^{(d-1)-\frac{d-1}p}\bnorm {\tilde f-\tau_{h e_i} \tilde f}_p.
			\end{align*}
			Since $\abs{h}=\abs{h e_i}_\infty$, the assertion follows from $\mu$ being an $\Lp$-modulus of $f$.
		\end{proof}

		This proposition implies that $\xis$ is polynomial-time reducible to $\xip$.
		Thus:
		\begin{thm}[efficientcy of integration]\label{resu:polynomial-time computablity of integration}
			$\xip$ renders the restriction of the integration operator from \cref{the integration operator} to $\LLL p\Omega\times \Omega^2$ polynomial-time computable.
		\end{thm}

		\begin{proof}
			Let $N$ be a natural number that bounds $1+\lb(\lambda(\Omega))+(d-1)\lb(\diam(\Omega))$.
			By the previous proposition the mapping padding the length of a $\xip$-name $\varphi$ of some function to have length $n\mapsto \length{\varphi}(n+N)$ is a polynomia-time translation to $\xis$.
			The assertion now follows from the minimality of the singular representation from Theorem~\ref{resu:minimality of the singular representation d}.
		\end{proof}

	\subsection{Equivalence to the Cauchy representation}\label{sec:sub:equivalence to the standard representation}
		
		Recall from Definition~\ref{def:metric spaces} that to obtain a Cauchy representation of $\LLL p{\Omega}$ it is sufficient to fix a dense subsequence.
		The obvious choice for this sequence are the dyadic step functions:
		Call a function a dyadic step function, if it is a dyadic linear combination of characteristic functions of sets of the form $[\str a,\str b]$ with $\bin{\str a}, \bin{\str b} \in \DD^d$.
		An enumeration of this set is cumbersome to write down, but there is one such that both the enumeration and its inverse are polynomial-time computable if a dyadic step function is identified with the list of the boxes on which it is non-zero and the corresponding values.
		The corresponding Cauchy representation is well established at least for investigating computability in $\LLL p{[0,1]}$ (cf. \cite{MR1005942,MR1724414,MR1694445} and many more).

		The goal of this section is to prove the following:

		\begin{thm}[equivalence to the Cauchy representation]\label{resu:equivalence to the standard representation}
			$\xip$ is computably equivalent to the Cauchy representation of $\LLL p\Omega$.
		\end{thm}

		One of the translations is easy and specified now.
		The other direction is more complicated and postponed until the end of the chapter.

		\begin{proof}[Proof that the Cauchy representation translates to $\xip$]\footnote{This proof was considerably simplified and strengthened thanks to a comment of an anonymous referee.}
			An oracle Turing machine that translates a name $\varphi$ of a function $f$ in the Cauchy representation into a $\xip$-name can be specified as follows:
			Given $\varphi$ as oracle and a string $\str c$ as input set $n:=\length{\str c}$.
			The machine obtains a valid value $\mu(n)$ of an $\Lp$-modulus of $f$ as follows:
			Let $f_{n+2}$ be the function encoded by $\varphi(n+2)$.
			Since this function is encoded as a list of the boxes it does not vanish on and its values on these boxes, the machine can obtain a bound $2^k$ on the number of boxes, $2^l$ of their diameters and $2^m$ of the values.
			Note that a dyadic step function that is defined as a linear combination of $2^k$ characteristic functions on sets of size $2^l$ can, when shifted by $y$, at most differ from the original function on a set of size $d\cdot\abs y_\infty\cdot2^{(d-1)l+k}$.
			Thus, since the difference can be majorized by $2^{m+1}$ on the set where it is nonzero, get
			\[ \sup_{\abs y_\infty \leq h}\|f_{n+2} - \tau_yf_{n+2}\|_p\leq 2^{m+1}\cdot\left(d\cdot h\cdot2^{(d-1)l}2^k\right)^{\frac 1p}. \]
			This means that $r:=\lceil p\rceil(d-1)l+k+\lceil\lb(d)\rceil+\lceil p\rceil(m+n+2)$ is a valid value of an $\Lp$-modulus of $f_{n+2}$ in $n+1$.
			Now, whenever $\abs y_\infty \leq 2^{-r}$ then
			\begin{align*}
				\| f - \tau_y f\|_p  & \leq \|f- f_{n+2}- \tau_y f + \tau_y  f_{n+2} \|_p +\norm{f_{n+2}-\tau_y f_{n+2}}_p \\
				& \leq 2\|f - f_{n+2}\|_p +\norm{f_{n+2} - \tau_y f_{n+2}}_p < 2^{-n}.
			\end{align*}
			Thus, $r$ is indeed a candidate for a value of an $\Lp$-modulus of $f$ in $n$.
			By repeating the procedure for all values of $n$ smaller than $\length{\str c}$ and increasing $r$ if necessary, the machine computes a value $\mu(\length{\str c})$ of an $\Lp$-modulus of $f$.

			Next, the machine checks if the input string is of the form $\str c = \langle \str a, \str b,\sdone^n\rangle$.
			If it is, it computes approximations of the integrals by returning the integrals of a dyadic step function which approximates the function accurately enough in $\Lp$.
			Before returning it, it pads the encoding of that approximation to length at least $\mu(\length{\str c})$.
			If the input string $\str c$ is not of the form $\langle \str a,\str b,\sdone^{n}\rangle$, the machine returns $\sdone^{\mu(\length{\str c})}$ (in this case only the length is relevant).
		\end{proof}
		The above proof can be checked to actually construct a polynomial-time reduction.

		The basic idea for the other direction is to approximate the function from $\LLL p\Omega$ by step functions where the values are the integrals over boxes.	
		For easier notation write
		\[ \cell x m :=  x +[-2^{-m-1},2^{-m-1}]^d. \]
		Note, that the Lebesgue measure of these sets is given by $\lambda(\cell xm) = 2^{-dm}$.

		\begin{defi}\label{def:continuous approximations}
			Let $f\in \Lp$ be a function.
			Define the \demph{sequence $(f_m)_{m\in\NN}$ of continuous approximations} to $f$ by
			\[ f_m(x) := 2^{dm}\int_{\cell xm}\tilde f \dd\lambda. \]
		\end{defi}

		A modulus of continuity of $f_m$ can be obtained from an $\Lp$-modulus of $f$:
		\begin{lem}[continuity]\label{resu:modulus of the continuous approximations}
			Whenever $\mu$ is an $\Lp$-modulus of $f\in \Lp$, the function $n\mapsto\mu(n+\lceil\frac dp\rceil m)$ is a modulus of continuity of $f_m$.
		\end{lem}

		\begin{proof}
			Use the version of Hölder's inequality from Corollary~\ref{cor:hoelder} to conclude
			\begin{align*}
				\abs{f_m(x)-f_m(y)} 
				& \leq 2^{dm} \int_{\cell xm} \big|\tilde f(t) - \tilde f(t-(x-y))\big| \dd t \leq 2^{\frac dpm}\bnorm{\tilde f-\tau_{x-y}\tilde f}_p.
			\end{align*}
			From this the assertion is obvious.
		\end{proof}

		How good an approximation $f_m$ is to $f$ can be read off from an $\Lp$-modulus of $f$:

		\begin{lem}[approximation]\label{resu:continuous approximations}
			Let $\mu$ be an $\Lp$-modulus of $f$.
			Then $\|\tilde f-f_{\mu(n)}\|_p < 2^{-n}$.
		\end{lem}

		\begin{proof}
			Using $\int_{\cell 0m} 2^{dm} \dd\lambda = 1$ one sees that
			\begin{align*}
				\|\tilde f-f_{m}\|_p^p & \leq \int_{\RR^d} \left|\tilde f(s) - 2^{dm}\int_{\cell 0m}\tilde f(t+s) \dd t\right|^p \dd s \\
				& \leq 2^{dmp}\int_{\RR^d} \left( \int_{\cell 0m} |\tilde f(s)-\tilde f(t+s)| \dd t\right)^p \dd s.
			\end{align*}
			Use the version of Hölder's inequality from Corollary~\ref{cor:hoelder} and Fubini to get
			\[ \|\tilde f-f_{m}\|_p^p \leq 2^{dm}\int_{\cell 0m}\int_{\RR^d}  |\tilde f(s)-\tilde f(t+s)|^p \dd s \dd t. \]
			Set $m:=\mu(n)$ and use that $\mu$ is an $\Lp$-modulus to see that
			\[ \|\tilde f-f_{m}\|_p < \left(2^{dm}\int_{\cell 0m} 2^{-pn} \dd t\right)^{\frac 1p} = 2^{-n}, \]
			which proves the assertion.
		\end{proof}

		We are now prepared to prove the second half of the main theorem of this section:
		\begin{prop}
			$\xip$ computably translates to the Cauchy representation of $\Lp(\Omega)$.
		\end{prop}

		\begin{proof}[Proof (also of Theorem~\ref{resu:equivalence to the standard representation})]\label{proof:second part of equivalence}
			Let $\varphi$ be a $\xip$-name of $f\in\LLL p{\Omega}$.
			Set $C:=\lceil\lb(\lambda(\Omega))\rceil$.
			For any $z\in\DD^d$ fix some binary encoding $\str a_z\in \albe^*$ (for example the unique canceled encoding).
			Furthermore, let $e$ be the constant one vector $(1,\ldots,1)$ and set
			\[ d_{z,k,N}:= 2^{dN}\bin{\varphi(\langle\str a_{z-2^{-N-1}e},\str a_{z+2^{-N-1}e},\sdone^{k}\rangle)}\quad\text{and}\quad \mu:=\flength{\varphi}. \]
			Thus, $d_{z,k,n}$ is a $2^{-k}$-approximation to the integral of $f$ over $\cell zN$ therefore also an approximation to the value of $f_N$ in $z$.

			Consider the step function
			\[ F_{k,N,M}:= \sum_{z\in \DD_M} d_{z,k,N}\chi_{\cell zM}, \]
			where $\DD_M$ denotes the set of $z\in\DD^d$ such that each component is of the form $\frac{m}{2^M}$ and such that $\cell zM\cap \Omega\neq \emptyset$.
			Since $\Omega$ is bounded there is a constant $D$ such that $\# \DD_M \leq 2^{dM+D}$.

			Obviously, the step function $F_{k,N,M}$ can be uniformly computed from the name $\varphi$ and the constants $k,N$ and $M$.
			To see how to choose $k$, $N$ and $M$ write
			\begin{align}
				\|f-F_{k,N,M}\|_p & \leq
				\|f-f_{N}\|_p + \norm{f_N-F_{k,N,M}}_p
			\end{align}

			By the approximation property of $f_N$ from Lemma~\ref{resu:continuous approximations} for the first summand to be smaller than $2^{-n}$, $N$ should be chosen $\mu(n+1)$.
			For the second summand, note that each $x\in\Omega$ is $2^{-M}$ close to some $z\in \DD_M$ and that for these $z$
			\[ \abs{f_N(z)-F_{k,N,M}(z)}<2^{dN-k}. \]
			Choosing $M:=\mu(n+\lceil d(N+C)/p \rceil+2)$ and $k:= dN+n+2$, using the modulus of continuity of $f_N$ from Lemma~\ref{resu:modulus of the continuous approximations} and that $F_{k,N,M}$ is piecewise constant obtain
			\[ \norm{f_N-F_{k,N,M}}_p \leq\lambda(\Omega)^{\frac 1p} \norm{f_N - F_{k,N,M}}_\infty < 2^{-n-1} \]
			Summing up, the result is smaller than $2^{-n}$.
		\end{proof}

		The Cauchy representation of $\Lp$ is continuous.
		Thus, the above proves that $\xip$ is a continuous mapping.
		Whenever $p$ is computable, the $\Lp$-norm is computable with respect to the Cauchy representation, and therefore also with respect to $\xip$.
		The above translation does not run in polynomial time as it accesses the oracle an exponential number of times.
		That no polynomial-time reduction exists can be seen from the results of the last chapter: With respect to $\xip$ the norm is not polynomial time computable.

		It is not to difficult to see that a minimality result like the ones for the representation of continuous functions (Theorem~\ref{resu:minimality of the standard representation}) and the singular representation (Theorem~\ref{resu:minimality of the singular representation d}) cannot be proven for this representation.
		This remains true if only the continuous
                representations are required to polynomial-time
                reduce.

	\section[Sobolev spaces]{Sobolev spaces}\label{sec:sobolev spaces}
		
		This chapter only considers the simplest domain $\Omega=[0,1]$.
		To simplify notation the domain is often omitted.
		For a function $f\in \LLL 1{[0,1]}$ a function $f'\in \LLL1{[0,1]}$ is called a \demph{weak derivative}, if for any $g\in \C^\infty([0,1])$ with $g(0) = 0 = g(1)$ it holds that
		\[ \int_{[0,1]} f g' \dd\lambda = - \int_{[0,1]} f' g \dd\lambda. \]
		Recall that, if it exists, the weak derivative of a function is uniquely determined (as an element of $\Lone$).
		Furthermore, if an element of $\Lone$ allows a weak derivative $f'$, then there is a continuous representative $f$ that fulfills
		\[ f(y)-f(x) = \int_{[x,y]} f'\dd\lambda. \]
		While in higher dimensions weakly differentiable functions may have singularities, in one dimension they are continuous.
		The following refinement of the continuity of a weakly differentiable function follows directly from Proposition~\ref{resu: from lp-modulus to singularity modulus}:
		\begin{lem}[differentiability and moduli]\label{resu:Lp-modulus of derivative to modulus of continuity}
			Whenever $f$ is weakly differentiable and $\mu$ is an $\Lp$-modulus of $f'$ then $n\mapsto \mu(n+1)$ is a modulus of continuity of $f$.
		\end{lem}

		The Sobolev space $W^{1,p}$ is defined as the set of functions from $\Lp$ that have a weak derivative which is also an $\Lp$-function.
		Sobolev spaces are of great importance in the theory of partial differential equations.
		It is well known that the Sobolev spaces can be characterized as spaces of functions with small $\Lp$-moduli (compare for instance \cite[Proposition 8.5]{MR2759829}).
		Since the named source uses different terminology and the result is stated for the whole space and not the unit interval we restate it and give a proof.
		\begin{lem}[small moduli]\label{resu:classes with small lp-moduli}
			The following are equivalent for $f\in\Lp$ with $1<p<\infty$:
			\begin{itemize}
				\item $f\in W^{1,p}$ and the continuous representative vanishes in $0$ and $1$.
				\item There is a $C\in \omega$ such that $n\mapsto n+C$ is an $\Lp$-modulus of $f$.
			\end{itemize}
			Furthermore, the constant $C$ can be chosen as any integer strictly larger than $\lb(\norm{f'}_p)$.
		\end{lem}

		\begin{proof}
			First assume that $f\in W^{1,p}$ and that the continuous representative vanishes at $0$ and $1$.
			In this case the extension $\tilde f$ to the whole real line by zero is continuous and its weak derivative is the extension of the weak derivative by zero.
			Use the version of Hölder's inequality from Corollary~\ref{cor:hoelder} to conclude
			\begin{align*}
				\norm{\tilde f - \tau_h\tilde f}_p & = \left(\int_{\RR} \abs{\int_x^{x+h} \tilde f'(t) \dd t}^p \dd x\right)^{\frac 1p} \leq h\left(\int_{\RR} \left(\int_0^1 \abs{\tilde f'(x+sh)} \dd s\right)^p \dd x\right)^{\frac 1p}\\
				& \stackrel{\ref{cor:hoelder}}\leq h\left(\int_{\RR} \int_0^1 \abs{\tilde f'(x+sh)}^p \dd s \dd x\right)^{\frac 1p} = h \norm{f'}_p.
			\end{align*}
			From this it is easy to see that $n+C$ is an $\Lp$-modulus of $f$ whenever $C$ is strictly larger than $\lb(\norm{f'}_p)$.

			For the other direction assume that $n+C$ is an $\Lp$-modulus of $f$.
			Recall that \cite[Proposition 8.5]{MR2759829} states that a function $g\in\LLL p{\RR}$ is an element of $W^{1,p}(\RR)$ if the inequality $\norm{g-\tau_h g}_p \leq D\abs{h}$ holds for all $h\in\RR$.
			$\tilde f$ fulfills this for $D:= 2^{2C+1}$:
			Given $h$ first check if there is a $n$ such that $2^{-\mu(n+1)}\leq \abs{h} < 2^{-\mu(n)}$.
			If so, then
			\[ \bnorm{\tilde f - \tau_h\tilde f}_p < 2^{-n} = 2^{-n+\mu(n+1)-\mu(n+1)} \leq 2^{C+1}\abs{h} \]
			If there is no such $n$, then $2^{-\mu(0)}\leq \abs{h}$ and using the norm bound from the $\Lp$-modulus by Lemma~\ref{resu:norm bound from modulus} conclude
			\[ \bnorm{\tilde f - \tau_h\tilde f}_p \leq 2\bnorm{\tilde f}_p < 2^{\mu(0)+1} \leq 2^{2C+1}\abs{h}. \]
			Thus, in any case
			\begin{equation}\label{eqn:linear modulus}\tag{h}
				\bnorm{\tilde f - \tau_h\tilde f}_p < 2^{2C+1} \abs{h}.
			\end{equation}
			It follows that the restriction of $\tilde f$ to $[0,1]$ is an element of the Sobolev space.
			Show that the continuous representative of $f$ vanishes on the boundary by contradiction:
			Assume $f(0)\neq 0$, w.l.o.g. $f(0)>0$.
			Then there exists some $\varepsilon$ and some interval $[0,\delta]$ such that $f(x) \geq \varepsilon$ for any $x\in[0,\delta]$.
			Set $h:= \min\big\{\delta,\big(\varepsilon2^{-2C-1}\big)^{1/(1-1/p)}\big\}$, then
			\[ \bnorm{\tilde f-\tau_h \tilde f}_p \geq \left(\int_0^h \abs{f}^p \dd\lambda \right)^{\frac 1p} \geq h^{\frac 1p} \varepsilon \geq  2^{2C+1} h = 2^{2C+1} \abs{h}, \]
			which contradicts \cref{eqn:linear modulus}.
			Therefore, $f$ vanishes in zero.
			The argument for the other end of the interval is identical.
		\end{proof}
		In the case $p=1$ one of the directions of the result fails: Characteristic functions of intervals have $n+1$ as $\Lone$-modulus while not being weakly differentiable.
		The other direction still holds true.

		In the remarks following Definition~\ref{def:modulus of continuity} of the modulus of continuity the corresponding class of functions was specified as the Lipschitz functions.
		In Proposition~\ref{resu:classes with small moduli} the class for the singularity modulus was proven to be $\LL\infty$.

	\subsection[Representing Sobolev spaces]{Representing \texorpdfstring{$W^{m,p}$}{Sobolev spaces}}

		Denote the $m$ times iterated weak derivative of a function $f$ by $f^{(m)}$.
		The \demph{Sobolev space $W^{m,p}$} is the space of all functions $f\in\Lp$ such that the weak derivatives $f',\ldots,f^{(m)}$ exist and are $\Lp$-functions.
		Equipped with the norm
		\begin{equation*}
			 \|f\|_{m,p}:= \sqrt[p]{\|f\|_p^p + \|f^{(m)}\|_p^p}
		\end{equation*}
		this space is a Banach space, and for $p=2$ a Hilbert space.
		In one dimension from $f^{(m)}$ being an $\Lp$ function it follows that $f^{(m-1)}$ is continuous.

		Recall the encoding $\bin\cdot$ of the dyadic numbers from the introduction and that for a length monotone string function $\length{\varphi}(\length{\str a})=\length{\varphi(\str a)}$.
		\begin{defi}\label{d:ximp}
			Define the \demph{second-order representation $\xi_{m,p}$ of $W^{m,p}$}: A length-monotone string function $\varphi$ is a $\xi_{m,p}$-name of $f\in W^{m,p}$ if for all strings $\str a, \str b$ such that $\bin{\str a},\bin{\str b}\in[0,1]$ and all $n\in \NN$
			\[ \abs{\int_{\bin{\str a}}^{\bin{\str b}} f\dd\lambda - \bin{\varphi(\langle \str a, \str b,\sdone^n\rangle)}} < 2^{-n}, \]
			and $\flength{\varphi}$ is an $\Lp$-modulus (see Definition~\ref{def:lp-moduli}) of the highest derivative $f^{(m)}$ of $f$.
		\end{defi}
		From now on always equip $W^{m,p}$ with the second-order representation $\xi_{m,p}$.
		The representations $\xip$ from Definition~\ref{d:xip} coincide with $\xi_{0,p}$, so no ambiguities arise.
		The space $\C([0,1])$ is always equipped with $\xic$.

	\subsection[One derivative]{\texorpdfstring{The space $W^{1,p}$}{One derivative}}
		
		Before investigating the space $W^{m,p}$ consider the simplest non-trivial case $m=1$.
		As a set $W^{1,p}$ is contained in $\Lp$.
		From the definition of the norm on $W^{1,p}$ it follows, that the inclusion mapping $W^{1,p}\hookrightarrow \Lp$ is continuous.
		\begin{thm}[Sobolev functions as $\Lp$-functions]\label{resu:sobolev functions as lp functions}
			The inclusion mapping $W^{1,p}\hookrightarrow \Lp$ is polynomial-time computable.
		\end{thm}

		For the proof it is necessary to obtain an $\Lp$-modulus of a function from a modulus of continuity and some extra information.
		The corresponding result is interesting on its own behalf.
		Therefore, we state it separately and in more generality than needed.
		\begin{lem}\label{resu:from cont. mod. to Lp-mod.}
			Let $\mu$ be a modulus of continuity of some function $f\in \C(\Omega)$ and let $\nu$ be an $\Lp$-modulus of the characteristic function of $\Omega$.
			Then an $\Lp$-modulus of $f$ is given by
			\[ \eta(n):= \max\left\{\mu(n+ \lceil\lb(\lambda(\Omega))\rceil+1),\nu(n+\lceil\lb(\norm{f}_\infty)\rceil+1)\right\}. \]
		\end{lem}
		\begin{proof}
			for sets $A$ and $B$ denote the symmetric difference by $A\Delta B := (A\cup B) \setminus (A\cap B)$.
			A function $\nu$ is an $\Lp$-modulus of the characteristic function of $\Omega$ if and only if from $|h|\leq 2^{-\nu(m)}$ it follows that $\lambda(\Omega\Delta(\Omega +h))^{1/p}< 2^{-m}$.
			Thus, for $\abs h \leq 2^{-\eta(n)}$
			\begin{align*}
				\|f-\tau_hf\|_p & \leq \|\chi_{\Omega\setminus(\Omega +h)\cup\Omega\setminus(\Omega-h)}f\|_p + \|\chi_{\Omega\cap(\Omega+h)}(f-\tau_h f)\|_p  \\
				& \leq \|f\|_\infty\cdot \lambda(\Omega\Delta(\Omega+h))^{\frac1p} + \left(\int_{\Omega\cap(\Omega+h))} |f-\tau_hf|^p\dd\lambda\right)^{\frac 1p} \\
				& < 2^{\lb(\norm f_\infty)}2^{-n-\lceil(\lb(\norm f_\infty))\rceil-1} + 2^{\frac{\lb(\lambda(\Omega))}{p}}2^{-n-\lceil\lb(\lambda(\Omega))\rceil-1} \leq 2^{-n}.
			\end{align*}
			Which proves the assertion.
		\end{proof}
		For $\Omega = [0,1]$ the characteristic function has $n\mapsto n+1$ as modulus and the previous result states that up to a bound on the norm, a modulus of continuity contains strictly more information about the function than an $\Lp$-modulus.

		\begin{proof}[Proof of Theorem~\ref{resu:sobolev functions as lp functions}]
			The following specifies an oracle Turing machine that transforms a $\xi_{1,p}$-name $\varphi$ of $f$ into a $\xi_{0,p}$-name of $f$:
			The approximations to the integrals for the $\xi_{0,p}$-name can be read from $\varphi$.
			To find the right length of the output, access to an $\Lp$-modulus of the function is needed.
			Since $\flength{\varphi}$ is an $\Lp$-modulus of $f'$, by Theorem~\ref{resu:Lp-modulus of derivative to modulus of continuity} $\mu(n):=\flength{\varphi}(n+1)$ is a modulus of continuity of $f$.
			Recall from Lemma~\ref{resu:from cont. mod. to Lp-mod.} that to obtain an $\Lp$-modulus of $f$ from a modulus of continuity of $f$ it suffices to have a bound on the supremum norm.
			By the mean value theorem for integration
			\[ \int_0^1 f\dd\lambda = f(y) \]
			for some $y\in [0,1]$.
			Let $\str a$ and $\str b$ be encodings of $0$ and $1$ as dyadic numbers.
			Then
			\begin{align*}
				\abs{f(y)} & \leq \abs{f(y) - \int_0^1 f\dd\lambda} + \abs{\int_0^1 f\dd\lambda-\bin{\varphi(\langle \str a,\str b,\varepsilon)}} + \abs{\bin{\varphi(\langle{\str a, \str b, \varepsilon})}} \leq \abs{\bin{\varphi(\langle{\str a, \str b, \varepsilon})}}+1.
			\end{align*}
			Choose some integer $Q$ such that $2^Q$ is a bound for $\abs{\bin{\varphi(\langle{\str a, \str b, \varepsilon})}}+1$.
			Bound the supremum norm of $f$ by using the modulus of continuity and the triangle inequality:
			Fix some $x\in[0,1]$ and set $x_i := x+(y-x)i2^{-\mu(0)}$, then $x_0=x$, $x_{2^{\mu(0)}} = y$ and $\abs{x_i-x_{i+1}}\leq 2^{-\mu(0)}$.
			Thus,
			\[ \abs{f(x)} \leq \sum_{i=0}^{2^{\mu(0)}-1} \abs{f\left(x_i\right)-f\left(x_{i+1}\right)}+\abs{f(y)} \leq 2^{\max\{\mu(0),Q\}+1}. \]
			Taking the supremum on both sides gives $\norm{f}_\infty\leq 2^{\max\{\mu(0),Q\}+1}$.

			Lemma~\ref{resu:from cont. mod. to Lp-mod.} now specifies an $\Lp$-modulus that can be computed from $\varphi$ in polynomial-time.
			Thus, the machine can pad the return values to an appropriate length.
		\end{proof}

		In one dimension, the Sobolev spaces consist of continuous functions and the inclusion mapping $W^{1,p}\hookrightarrow \C([0,1])$ is well known to be continuous (for $1<p\leq\infty$ it is compact).

		\begin{thm}[inclusion into continuous functions]\label{resu:Sobolev functions as continuous functions}
			The inclusion mapping $W^{1,p}\hookrightarrow \C([0,1])$ is polynomial-time computable.
		\end{thm}

		\begin{proof}
			Let $\varphi$ be a $\xi_{1,p}$-name of a function $f\in W^{1,p}$.
			Describe an oracle Turing machine that transforms this name into a $\xic$-name of $f$:
			Assume the machine is given some input $\str c$ and provided $\varphi$ as oracle.
			Note that by Theorem~\ref{resu:Lp-modulus of derivative to modulus of continuity} the mapping $\mu(n):=\length{\varphi}(n+1)$ is a modulus of continuity of the continuous representative of $f$.
			Therefore the necessary length of the return value is known.
			If the input is not of the form $\str c = \langle \str a, \sdone^n\rangle$, where $\str a$ is the encoding of some dyadic number $d\in[0,1]$ return a sufficiently long sequence of zeros.
			If it is of that form an approximation to $f(d)$ can be obtained as follows:
			By the mean value theorem
			\[ 2^{\mu(n+1)+1}\int_{d-2^{-\mu(n+1)}}^{d+2^{-\mu(n+1)}}f \dd\lambda = f(y) \]
			for some $y\in[d-2^{-\mu(n+1)},d+2^{-\mu(n+1)}]$ and therefore
			\[ \abs{f(d) - 2^{\mu(n+1)+1}\int_{d-2^{-\mu(n+1)}}^{d+2^{-\mu(n+1)}}f \dd\lambda} < 2^{-n-1}. \]
			Let $\str b^\pm$ denote encodings of $d\pm2^{-\mu(n+1)}$.
			Such encodings are easily obtained from $\str a$.
			Then $2^{\mu(n+1)+1}\bin{\varphi(\langle \str b^-, \str b^+, \sdone^{\mu(n+1)+n+2} \rangle)}$ is an approximation to $f(d)$ and (a sufficiently long encoding is) a valid return value.
		\end{proof}

		Note that this result does not imply the previous Theorem~\ref{resu:sobolev functions as lp functions}:
		Polynomial time computability of the restriction of the integration operator from \cref{the integration operator} is known to fail on $\C([0,1])$ (for instance \cite[Example 6h]{MR3377508}).
		On $\Lp$ on the other hand this operator is polynomial-time computable by Theorem~\ref{resu:polynomial-time computablity of integration}.
		Thus, the inclusion mapping $\C([0,1])\hookrightarrow \Lp$ is not polynomial-time computable.

		\begin{cor}[differentiation]
			The operator
			\[ \frac{d}{dx}: W^{1,p} \to \Lp,\quad f\mapsto f' \]
			is polynomial-time computable.
		\end{cor}

		\begin{proof}
			A given $\xi_{1,p}$-name $\varphi$ of a function $f\in W^{1,p}$ can be transformed into a $\xi_{0,p}$ name of $f'$ in polynomial-time as follows:			
			An $\LL p$-modulus is contained in the $\xi_{1,p}$-name.
			It remains to compute the integrals.
			By the previous theorem it is possible to obtain approximations to the values of $f$ on dyadic numbers.
			Using the formula
			\begin{align*}
				f(y)-f(x) = \int_x^{y} f'\dd\lambda
			\end{align*}
			and the triangle inequality these can be converted to approximations of the integrals.
		\end{proof}

	\subsection[Higher derivatives]{\texorpdfstring{The space $W^{m,p}$}{Higher derivatives}}

		Recall from Definition~\ref{d:ximp} that a name of a $W^{m,p}$ function contains information about the integrals of the function over dyadic intervals and an $\Lp$-modulus of the highest derivative of $f$.
		If $m>1$ it is not so easy to combine information contained in the $\Lp$-modulus and in the integrals of the function.
		The key is to iteratively apply the mean value theorem:

		\begin{lem}\label{resu:mean value}
			Whenever $f\in W^{m,p}$ and $(x_i)_{i\in\{1,\ldots,2^{m-1}\}}\subseteq [0,1]$ are of pairwise distance at least $2^{-m}$ such that $\abs{f(x_i)}\leq C$, then there exists a $z\in[0,1]$ such that $f^{(m-1)}(z) \leq 2^{m^2-1}C$.
		\end{lem}

		\begin{proof}
			Recursively for any $k<m$ construct a family of points $(x^k_i)_{i\in\{1,\ldots,2^{m-k-1}\}}$ of pairwise distance at least $2^{-m}$ such that $f^{(k)}(x^k_i)\leq 2^{k(m+1)}C$.

			The case $k=0$ is taken care of by the assumption.
			Now assume availability of a family $x^{k-1}_i$ as needed.
			Since $f^{(k-1)}$ is a continuously differentiable function whenever $k<m$, the mean value theorem states that for any $j\in\{1,\ldots,2^{m-k-1}\}$ there is some element $x^k_j\in[x^{k-1}_{2j-1},x^{k-1}_{2j}]$ such that
			\[ f^{(k)}(x^k_j) = \frac{f^{(k-1)}(x^{k-1}_{2j-1}) - f^{(k-1)}(x^{k-1}_{2j})}{x^{k-1}_{2j-1}-x^{k-1}_{2j}} \]
			and therefore
			\[ \abs{f^{(k)}(x^k_j)} \leq 2\cdot 2^{(k-1)(m+1)} C \cdot 2^m = 2^{k(m+1)} C. \]
			Obviously, the distance of the points will not decrease.

			Setting $k=m-1$ proves the lemma.
		\end{proof}

		\begin{prop}[some Sobolev embeddings]
			The inclusion mapping $W^{m,p}\hookrightarrow W^{m-1,p}$ is polynomial-time computable.
		\end{prop}

		\begin{proof}
			Let $\varphi$ be a $\xi_{m,p}$-name of a function $f\in W^{m,p}$.
			Compute the value of a $\xi_{m-1,p}$-name of $f$ on a string $\str a$ as follows:
			To get an $\Lp$-modulus of $f^{(m-1)}$ from the $\Lp$-modulus of $f^{(m)}$ use the previous Lemma:
			By Theorem~\ref{resu:Lp-modulus of derivative to modulus of continuity} the function $\mu(n):= \flength{\varphi}(n+1)$ is a modulus of continuity of $f^{(m)}$.
			Use the mean value theorem for integrals like in the proof of Theorem~\ref{resu:Sobolev functions as continuous functions} to produce a family of points and a constant $C$ that fulfill the assumptions of Lemma~\ref{resu:mean value}.
			The lemma provides an explicit bound for the values of $f^{(m-1)}$.
			Combine this with the modulus of continuity like at the end the proof of Theorem~\ref{resu:Sobolev functions as continuous functions} to get an integer bound $Q$ on $\lb(\norm{f^{(m-1)}}_\infty)$.
			By Lemma~\ref{resu:from cont. mod. to Lp-mod.}
			\[ n\mapsto \max\left\{\mu(n+1),n+Q+1\right\} \]
			is an $\Lp$-modulus of $f^{(m-1)}$.
			This function can be computed in polynomial-time and the padded return values of $\varphi$ are valid return values.
		\end{proof}
		The algorithm specified in this proof accesses the oracle about $2^m$ times.
		This does not lead to exponential time consumption as $m$ is fixed, however it might lead to large constants in the polynomials for the running time.

		The following Theorems can be proven by induction, where Theorems~\ref{resu:sobolev functions as lp functions} and \ref{resu:Sobolev functions as continuous functions} are the base cases and the previous proposition is the induction step.
		\begin{thm}\label{resu:higher sobolev functions as lp functions}
			The inclusion $W^{m,p}\hookrightarrow \Lp$ is polynomial-time computable.
		\end{thm}

		\begin{thm}\label{resu:higher sobolev functions as continuous functions}
			The inclusion $W^{m,p}\hookrightarrow \C([0,1])$ is polynomial-time computable.
		\end{thm}

		Finally consider the differentiation operator:
		\begin{cor}\label{resu:differentation}
			The $k$-wise differentiation operator
			\[ \frac{d^k}{dx^k}:W^{m,p} \to W^{m-k,p}, \quad f \mapsto f^{(k)} \]
			is polynomial-time computable for all $k\leq m$.
		\end{cor}

		\begin{proof}
			By the previous theorem obtain approximations to the values of $f$ on dyadic elements.
			By means of
			\[ f(x)-f(y) = \int_y^x f'\dd\lambda \]
			convert these into approximations of the integrals over $f'$.
			Iterate this process $k$-times to obtain approximations to the integrals over $f^{(k)}$.
		\end{proof}

\section[Motivating the use of the Lp-modulus]{Motivating the use of the \texorpdfstring{$\Lp$}{Lp}-modulus}\label{sec:motivationg the use of the lp-modulus}

	This last chapter provides evidence that the $\Lp$-modulus is far of from an arbitrary choice as the length parameter for a representation of $\Lp$.
	The origin of the notion of an $\Lp$-modulus as replacement for the modulus of continuity for $\Lp$-spaces is a classification theorem of the compact subsets of $\Lp$-spaces.
	Before considering the $\Lp$ case recall the following well known theorem from analysis:
	\begin{thm}[\aaT]\label{resu:aaT classical}
		A subset of $\C([0,1])$ is relatively compact if and only if it is bounded and equicontinuous.
	\end{thm}
	Equicontinuity of a subset of $\C([0,1])$ is equivalent to the existence of a common modulus of continuity of all of the elements.
	Thus, this theorem provides a direct connection between compactness of a set of functions and their moduli of continuity.

	A similar theorem is known for $\Lp$-spaces, where equicontinuity is replaced by the existence of a common $\Lp$-modulus.
	\begin{thm}[\fkT]\label{resu:fkT classical}
		Let $1\leq p <\infty$.
		A subset $F$ of $\LLL p{[0,1]}$ is relatively compact if and only if it is bounded and there is a function $\mu:\omega\to\omega$ that is an $\Lp$-modulus of all of the functions from $F$.
	\end{thm}
	These statements are only qualitative.
	Quantitative refinements can be related optimality results for second-order representations of these spaces.
	For these refinements a notion of \lq size\rq\ for compact sets is needed.

	\subsection{Metric entropies and spanning bounds}\label{sec:sub:metric entropy}

		It is well known that in a complete metric space a subset is relatively compact if and only if it is totally bounded.
		The following notion is a straight forward quantification of total boundedness and can be used to measure the \lq size\rq\ of compact subsets of metric spaces.
		It was first considered in \cite{MR0112032}, where many of the names we use originate.
		A comprehensive overview can be found in \cite{lorentz1966}.
		These notions have been applied to computable analysis before \cite{MR1952428} and were also used in \cite{MR2130066}.
		For the following let $M$ be a metric space and $d$ its metric.

		\begin{defi}\label{def:metric entropy and size}
			A function $\nu:\omega \to \omega$ is called \demph{modulus of total boundedness} of a subset $K$ of $M$, if for any $n\in\omega$ there are $2^{\nu(n)}$ balls of radius $2^{-n}$ that cover $K$.
			The smallest modulus of total boundedness is called the \demph{metric entropy} or \demph{size} of the set and denoted by $|K|$.
		\end{defi}

		Thus
		\[ \size{K}(n) = \min\{k\in\omega\mid K \text{ can be covered by }2^k\text{ balls of radius }2^{-n}\}. \]
		Like the smallest modulus of continuity of a function, the metric entropy of a set is usually hard to get hold of.
		Moduli of total boundedness as upper bounds can more often be chosen computable.
		In a complete metric space a closed set permits a metric entropy if and only if it is compact.

		A modulus of total boundedness is an upper bound on the size of a compact set.
		For providing lower bounds, another notion is more convenient.

		\begin{defi}\label{def:spanning bound}
			A function $\eta:\omega\to\omega$ is called a \demph{spanning bound} of a subset $K\subseteq M$, if for any $n$ there exist elements $x_1,\ldots,x_{2^{\eta(n)}}$ such that
			\[ i\neq j \quad \Rightarrow\quad d(x_i,x_j)\geq 2^{-n+1}. \]
			If there is a biggest spanning bound, it is called the \demph{capacity} of $K$ and denoted by $\capa{K}$.
		\end{defi}

		The condition on the $x_i$ in the definition can be read as \lq the $2^{-n}$-balls around the points are disjoint\rq.
		There is a biggest spanning bound if and only if the set $K$ is relatively compact.
		The following is straight forward to verify:

		\begin{prop}\label{resu:lower bound}
			Let $K\subseteq M$ be a subset, $\nu$ be a metric entropy of $K$ and $\eta$ a spanning bound.
			Then $\eta(n)\leq \nu(n)$, and furthermore $\size{K}(n)\leq \capa{K}(n+1)+1$.
		\end{prop}
		A proof can for instance be found in \cite[Theorem IV]{MR0112032}, however, the result presented there is a little sharper since rounding to powers of two is avoided.

		This implies comparability of the size and the capacity in the sense that
		\[ \capa{K}(n) \leq \size{K}(n)\leq \capa{K}(n+1)+1. \]
		This paper uses spanning bounds to provide lower bounds to the size of sets.
		The capacity is not mentioned again.

	\subsection{Connecting metric entropy and complexity}\label{sec:connecting metric entropy and complexity}

		For sake of completeness this chapter states the result connecting the metric entropy to computational complexity.
		Some of the results presented in this chapter are in a slightly different form contained in \cite{CIE2016}.
		There are several examples of similar observations \cite{MR1911553,MR3239272}.

		As already mentioned, neither the norm nor the metric of $\C([0,1])$ or $\Lp$ are polynomial-time computable with respect to the representations considered in this paper.
		It is, however, possible to give resource bounds:
		The operations are polynomial space computable.
		Since space restricted computation in presence of oracles is tricky (compare \cite{MR973445,MR3259646}), this paper considers the exponential-time computations instead.
		\begin{defi}\label{def:exponential time computability}
			An oracle Turing machine $M^?$ is said to run in \demph{exponential-time} if there are constants $A,B,C\in\omega$ such that the computation $M^\psi(\str a)$ of $M^?$ with oracle $\varphi$ on input $\str a$ terminates after at most $2^{A\cdot \flength{\varphi}(\length{\str a}+B)+C\length{\str a}}$ steps.
		\end{defi}
		Call a function between represented spaces \demph{exponential-time computable} if it has an exponential-time computable realizer (compare \Cref{sec:sub:second-order complexity theory}).
		This notion of exponential-time computability is highly adapted for the concrete application at hand and quite restrictive:
		While an exponential running time is only bounded by a second-order polynomial if it is constant, not all second-order polynomials can be bounded by exponentials.
		As soon as there are iterations of the first order argument, no such bound exists.

		Furthermore, the following notions are needed:
		\begin{defi}
			A function $l:\omega\to\omega$ is called a \demph{length} of a second-order representation if each element of the represented space has a name of length at most $l$.
		\end{defi}
		It can be proven that an open representation of a compact space always has a length.
		The standard representations of the real numbers, the continuous functions, $\Lp$-spaces and integrable functions discussed in this paper so far do not have a  length.
		But since all of them are open mappings, their range restrictions to compact subsets do have a length.
		The proof of Theorem~\ref{resu:minimality of the singular representation d} makes use of the finite length of the restriction of the standard representation of the reals to any bounded set.

		\begin{thm}\label{resu:metric entropy and complexity}
			Let $M$ be a compact metric space of at least linear metric entropy.
			\begin{enumerate}
				\item Assume that there exists a representation of $M$ of length $l$ such that the metric is computable in exponential-time.
				Then there exist some $A,B\in\omega$ such that
				\[ \size{M}(n)\leq 2^{Al(n)+B}. \]
				\label{resu:metric entropy and complexity item:bounded entropy}
				\item Let $l:\omega\to\omega$ be monotone such that $\size{M}(n)\leq2^{l(n)}$.
				Then there exists a representation $\xi$ of $M$ that has length $l$ such that the metric is computable in exponential-time.\label{resu:metric entropy and complexity item:bounded time}
			\end{enumerate}
		\end{thm}

		\begin{proof}[Sketch of the proof]
			To prove \Cref{resu:metric entropy and complexity item:bounded entropy} use the folklore fact that a running time restricts the access a machine has to the oracles (compare for instance \cite{MR1911553}).
			Make this a quantitative statement by bounding the number of possible communication sequences.
			From this obtain a bound on the number of pairs $\langle \varphi,\psi\rangle$ that can be distinguished in a computation of the norm up to precision $2^{-n}$.
			This leads to a bound on the size of the set.

			To prove \Cref{resu:metric entropy and complexity item:bounded time} let $l$ be such that $\size M(n)\leq 2^{l(n)}$.
			Then there exists a sequence such that the balls around the first $2^{2^{l(n)}+\lceil\lb(n+1)\rceil}$ elements cover $M$.
			Consider the Cauchy representation of $M$ with respect to this sequence (according to Definition~\ref{def:metric spaces}).
			This representation has length $2^{l(n)}+\lceil{\lb(n+1)}\rceil$.
			Add an oracle for the function $(n,m)\mapsto d(x_n,x_m)$ to each name, truncate the representation and flatten it to obtain one of constant length.
			Now pad the length of each name to be $l(n)$.
			This is the desired representation.
		\end{proof}

		\begin{exa}[Lipschitz functions]
			Consider as $K$ the set of Lipschitz functions that vanish in $0$ and have a common Lipschitz constant $2^L$.
			This set is compact by the \aaT Theorem from Theorem~\ref{resu:aaT classical}.
			It is easy to verify that $\size{K}(n)=\lceil\lb(3)2^{n+L}\rceil$ (see for instance \cite{MR1952428}).
			Let $\xi$ be an open representation of this set such that the metric is computable in exponential time.
			As an open representation of a compact set, this representation has a length $l$.
			From the first item of the previous theorem it follows that there exists constants $A,B\in\omega$ such that
			\[ \size{K}(n) =\lceil\lb(3)2^{n+L}\rceil \leq 2^{Al(n)+B} \]
			and therefore $l(n) \geq (n+L-B)/A$, i.e. $\xi$ has at least linear length.

			The second item of the previous theorem specifies a representation of linear length of $K$ that renders the metric exponential-time computable.
			In this special case we already knew that such a representation exists:
			The range restriction of the standard representation $\xic$ of continuous functions to $K$ has length $n\mapsto n+C$, where $C$ depends on the Lipschitz constant, and it renders the metric exponential-time computable.
		\end{exa}

		The rest of this chapter aims to generalize the above example by replacing the set of Lipschitz functions with more general compact sets.
		The compact subsets are fully classified by the \aaT and \fkT Theorems.
		By specifying the size of these compact sets it is possible to verify that  $\xic$ and $\xip$ have the minimal length for any representation that renders the metric exponential-time computable on these.
		This justifies the modulus of continuity and the $\Lp$-modulus as the right parameters for these function spaces.

	\subsection{\aaT and \fkT}

		The quantitative refinements of both the \aaT and \fkT Theorems have been investigated before in different contexts:
		There has been extensive work on these topics in approximation theory (for instance \cite{MR0112032} or \cite{MR1262128}).
		These results can, however, not straight forwardly be transfered to the context of this paper.
		In approximation theory the notion of moduli considered differs by convention.
		As a result, the theorems usually talk about the inverse modulus instead of the modulus itself.
		Furthermore, the results are often only stated or valid for small moduli.
		A very popular class is for instance the class corresponding to Hölder continuous functions for the continuous functions.

		There have been some attempts to apply the results to computable metric spaces: Most prominently \cite{MR1952428}.
		However, the results seem rather restricted.

		\begin{defi}\label{def:aaT sets}
			Define the family of \demph{\aaT-sets} $(K^\infty_{l,C})_{l\in \omega^\omega,C\in\omega}$ in $\C([0,1])$ by
			\[ K^\infty_{l,C} := \left\{f\in\C([0,1])\mid \text{$f$ has $l$ as modulus and $\|f\|_\infty\leq 2^{C}$}\right\}. \]
		\end{defi}

		The classical \aaT Theorem~\ref{resu:aaT classical} states that a set of functions is relatively compact if and only if it is contained in some $K^\infty_{l,C}$.
		The quantitative refinement of that statement can be found in \cite{MR1262128}.
		The proof given there is very similar to the one here.
		A restricted version is also proven in \cite{MR1952428}.

		\begin{thm}[\aaT]\label{resu:aaT}
			A set $K\subseteq C([0,1])$ is relatively compact if and only if it is contained in $K^\infty_{l,C}$ for some $l,C$.
			Furthermore:
			\[ 2^{l(\max\{n-2,0\})+\min\{n-2,0\}}+n+C\leq \size{K^\infty_{l,C}}(n)\leq 2^{l(n)+1} + n + C +2. \]
		\end{thm}

		\begin{proof}
			The first assertion follows from the classical \aaT Theorem~\ref{resu:aaT classical} and the proof is not repeated here.
			To provide the upper bound on the size of $K^\infty_{l,C}$ fix some $n\in\omega$.
			A collection of balls of size $2^{-n}$ that cover $K^\infty_{l,C}$ can be constructed as follows:
			Consider the index set
			\[ I := \{-2^{n+C},\ldots,2^{n+C}\}\times\{0,1,-1\}^{2^{l(n)}}. \]
			For $\sigma =(\sigma_0,\sigma_1,\ldots,\sigma_{2^{l(n)}})\in I$ define a piecewise linear function $f_\sigma:[0,1]\to \RR$ by
			\[ f_\sigma(x) =
				\begin{cases}
					\sigma_0 2^{-n} & \text{ if } x=0 \\
			 		2^{-n}\left(\sum_{i=0}^{j-1} \sigma_i + \sigma_j(2^{l(n)}x - j)\right) & \text{ if } x\in\left(\frac{j-1}{2^{l(n)}},\frac{j}{2^{l(n)}}\right]\text{ for some $j\in\NN$.}
			 	\end{cases}
			\]
			\begin{minipage}{.55\textwidth}
				The $2^{-n}$-balls centered at $f_\sigma$ cover $K^\infty_{l,C}$ and
				\[ \#I=(2^{n+C+1}+1)3^{2^{l(n)}}\leq 2^{2^{l(n)+1} + n + C + 2}. \]
				Since $n$ was arbitrary, the right hand side is an upper bound on the size of $K^\infty_{l,C}$.

				To establish the lower bound replace $f_\sigma$ with the function $g_\sigma$ that may or may not have a bump of size $2^{-n-1}$ in the $i$-th interval of size $2^{-l(n)}$.
				The extra condition that a modulus of continuity has to be strictly increasing when non-zero implies that whenever the value of $l$ on $n$ allows the function to vary by $2^{-n}$ over an interval of length $2^{-l(n)}$ the subsequent values of $l$ will not disallow this behavior.
				Thus $g_\sigma\in K^\infty_{l,C}$.
				For any two different elements $\sigma$ and $\sigma'$ of $I$ it holds that

				\vspace{.1cm}
			\end{minipage}
			\hfill
			\begin{minipage}{.5\textwidth}
				\centering
				\begin{tikzpicture}
					\draw[->] (0,-1.5) -- (0,3);
					\draw[->] (0,0) -- (4.5,0);
					\draw (-.1,.0) -- (.1,.0);
					\draw (-.1,.5) -- (.1,.5);
					\draw (-.1,1) -- (.1,1);
					\draw (-.1,1.5) -- (.1,1.5);
					\draw (-.1,2) -- (.1,2);
					\draw (-.1,2.5) -- (.1,2.5);
					\draw (-.1,-.5) -- (.1,-.5);
					\draw (-.1,-1) -- (.1,-1);
					\node at (-.3,2.5) {$2^{C}$};
					\node at (-.4,.5) {$2^{-n}$};
					\node at (-.55,1.5) {$\frac{\sigma_0}{2^{n}}$};
					\draw[color=blue] (0,1.5) -- (.5,1) -- (1,1.5) -- (1.5,1.5);
					\node at (.5,-.25) {$2^{-l(n)}$};
					\draw (.5,.1) -- (.5,-.1);
					\draw (1,.1) -- (1,-.1);
					\draw (1.5,.1) -- (1.5,-.1);
					\node at (2.25,.1) {\ldots};
					\draw (3,.1) -- (3,-.1);
					\draw (3.5,.1) -- (3.5,-.1);
					\draw (4,.1) -- (4,-.1);
					\draw[color=blue] (3,-1) -- (3.5,-1) -- (4,-.5);
				\end{tikzpicture}
				\hspace{.2cm}
				
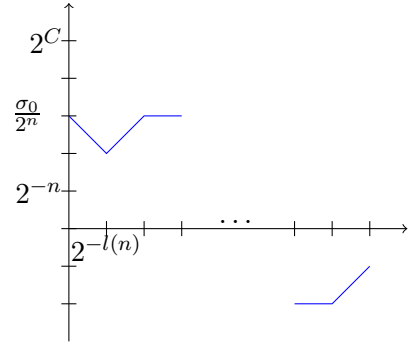
\captionof{figure}{The function $f_{\sigma}$ for $\sigma= (\sigma_0,-1,1,0,\ldots,0,1)$.}
			\end{minipage}
			$\|g_\sigma - g_{\sigma'}\|_\infty\geq 2^{-n-1}$.
			Thus, $n+2\mapsto 2^{l(n)} + n + C + 2$	is			a spanning bound in the sense of Definition~\ref{def:spanning bound}.
			Since any spanning bound of a set has to be smaller than its size by Proposition~\ref{resu:lower bound}, the lower bound on $\bsize{K^\infty_{l,C}}$ follows.
		\end{proof}

		For $\Lp$-spaces replace the sets $K^\infty_{l,C}$ by the following sets:
		\begin{defi}\label{def:fkT sets}
			Define the Family of \demph{\fkT-sets} $(K^p_{l})_{l\in \omega^\omega}$ in $\LLL p{[0,1]}$ by
			\[ K^p_l := \{f\in\Lp([0,1])\mid \text{ $f$ has $l$ as $\Lp$-modulus}\}. \]
		\end{defi}
		In this case there is no need to include an upper bound to the norm, since this bound can be extracted from an $\Lp$-modulus by Lemma~\ref{resu:norm bound from modulus}.
		\begin{thm}[\fkT]\label{resu:fkT}
			A set $K\subseteq \LLL p{[0,1]}$ is relatively compact if and only if it is contained in $K^p_l$ for some $l$.
			There exists a second-order polynomial $P$ such that
			\[ \size{K^p_l}(n)\leq 2^{P(l,n)}. \]
			Whenever $n\geq 3$ and $l(n-3)\geq 9$ it holds that $2^{l(n-3)-4}-1 \leq \size{K^p_l}(n)$.
		\end{thm}

		\begin{proof}[Proof of the upper bound]
			Using the lemmas from \Cref{sec:sub:equivalence to the standard representation}:
			By Lemma~\ref{resu:continuous approximations} any function $f\in K^p_l$ has a continuous function $f_{l(n+1)}$ in its $2^{-n-1}$ neighborhood.
			Lemma~\ref{resu:modulus of the continuous approximations} guarantees $f_{l(n+1)}\in K^\infty_{l'}$ for  $l'(m):=m+1+\left\lceil d/p\right\rceil l(n+1)$.
			The \aaT Theorem~\ref{resu:aaT} proves
			\[ \size{K^\infty_{l'}}(n)\leq 2^{l'(n)+1} + n + \lceil\lb(\bnorm{f_{l(n+1)}}_\infty\rceil +2 \leq 2^{l\left(n+2+\left\lceil\frac dp\right\rceil l(n+1)\right) +1} + n + l\left(1+\left\lceil\frac dp\right\rceil l(1)\right) +2. \]
			The balls of radius $2^{-n-1}$ in supremum norm are included in the balls in $\LL p$-norm.
			Therefore, the $2^{-n}$ balls in $\Lp$ around the same centers cover $K^p_l$.
		\end{proof}

		To find a lower bound, we use the technique Lorentz used in \cite{lorentz1966} for the prove of his Lemma 8.
		Namely we use the following lemma from coding theory:

		\begin{lem}\label{resu:coding theory}
			For any natural number $N \geq 500$ and $M<\frac N3$ there exists a set $I$ of binary strings of length $N$ that differ pairwise in at least $M$ places and such that
			\[ \#{I}=\left\lfloor2^{\frac N{16}-1}\right\rfloor. \]
		\end{lem}
		\begin{proof}
			Prove the stronger statement that there is a set $I$ of strings of length $N\geq 500$ that differ in at least $M= \frac N3$ bits whenever
			\[ \#{I}\leq\Big\lfloor2^{\frac N{16}}\Big(\frac e\pi2^{-\frac32}+1\Big)^{-1}\Big\rfloor. \]	

			Proceed by induction over the size $\#I$ of the set $I$.
			For $\# I=2$ choose the constant zero string and the constant one string.
			Now assume that $I$ is a set of strings that differ pairwise in at least $M$ elements and that has strictly less elements than the number specified above.
			Use Stirling's Formula to estimate the number of strings that differ in less than $M$ digits from one of the elements of $I$:
			\begin{align*}
				\#I \sum_{i=0}^M \binom Ni &\stackrel{M\leq \frac N2}\leq \#I(M\binom NM +1) \stackrel{\text{Stirling}}\leq \#I \left(\frac{Me}{2\pi} \frac{N^{N+\frac12}}{M^{M+\frac12}(N-M)^{N-M +\frac12}}+1\right) \\
				& \stackrel{M=\frac N3}= \#I\left(\frac{\sqrt Ne}{6\pi}3^{N+1}2^{-\frac23 N-\frac 12}+1\right) 
				\stackrel{N\geq 500}\leq \#I \left(\frac e{2^{\frac32}\pi}+1\right)2^{-\frac N{16}}2^N.
			\end{align*}
			By induction hypothesis the right hand side is strictly smaller than $2^N$.
			Since the left hand side is an integer it is at most $2^N-1$.
			Thus, at least one of the $2^N$ strings of length $N$ does not lie in the union of these sets and can be added to $I$ to increase its size by one.
		\end{proof}
		\begin{rem}\label{rem:coding theory}
			From coding theory it is known that these bounds are not optimal.
			In particular the assumption $N\geq 500$ can be removed.
			See for instance \cite{MR1948693}.
		\end{rem}

		\begin{proof}[Proof of the lower bound in Theorem~\ref{resu:fkT}]
			Fix some $n\in \omega$.
			The assumption $l(n-3) \geq 9$ guarantees that Lemma~\ref{resu:coding theory} can be applied with $N:=2^{l(n-3)}$ and $M:=2^{l(n-3)-2} = \frac N4$ to find a set $I$ of strings of length $2^{l(n-3)}$ such that the elements differ pairwise in at least $2^{l(n-3)-2}$ digits and $\#I = 2^{2^{l(n-3)-4}-1}$.
			Consider the functions
			\[ f:\RR\to [0,1],\quad x\mapsto \max\big\{0,1-2\bsize{x-\frac12}\big\},\quad
			 	w(n) := 2^{-l(n-3)}\text{ and }
				h(n) := (p+1)^{\frac1p}2^{-n+1}.
			\]
			\noindent
			\begin{minipage}{.55\textwidth}
				For each $\sigma \in I$ define a function $f_\sigma$ by
				\[ f_\sigma := h(n)\sum_{i=1}^{2^{l(n-3)}} \sigma_i f\left(\frac {x-iw(n)}{w(n)}\right). \]
				That is: Divide $[0,1]$ into intervals of width $w(n)$ and consider the set of functions that may or may not have a hat of height $h(n)$ in each of the intervals (see \Cref{fig:fsigma two}).
				Since at most one hat is put in each interval for each string $\sigma$ and $x\in[0,1]$ it is true that almost everywhere $f_\sigma(x)<h(n)$ and therefore $\norm {f_\sigma}_p < h(n)$.
				For the weak derivative of $f_\sigma$ it holds that $\norm {f'_\sigma}_\infty \leq {2h(n)/w(n)}$.
			\end{minipage}
			\begin{minipage}{.44\textwidth}
				\vspace{-.5cm}
				\centering
				\begin{tikzpicture}
					\draw[->] (0,0) -- (0,2);
					\draw[->] (0,0) -- (4.5,0);
					\draw (-.1,1.5) -- (.1,1.5);
					\node at (-.4,1.5) {$h(n)$};
					\draw[thick,color=blue] (0,0) -- (.5,0);
					\draw (.5,.1) -- (.5,-.1);
					\node at (.5,-.25) {\tiny$w(n)$};
					\node at (1.25,.1) {\ldots};
					\draw (2,.1) -- (2,-.1);
					\draw (2.5,.1) -- (2.5,-.1);
					\node at (2.75,-.25) {\tiny{$iw(n)$}};
					\node at (3.25,.1) {\ldots};				
					\draw (4,.1) -- (4,-.1);
					\draw[color=blue] (2,0) -- (2.25,1.5) -- (2.5,0);
				\end{tikzpicture}
				
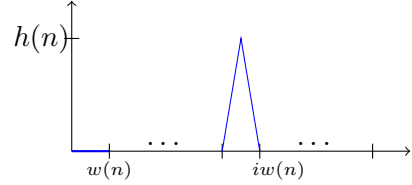
\captionof{figure}{$f_{\sigma}$ for $\sigma$ with $\sigma_1=0$ and $\sigma_i=1$.}\label{fig:fsigma two}
				\vspace{-.1cm}
			\end{minipage}
			
			To obtain the spanning bound prove that the $f_\sigma$ are elements of $K^p_l$ of pairwise distance more than $2^{-n}$:					
			To show that these functions are elements of $K^p_l$, claim that the smallest valid modulus function $\mu$ such that $\mu(n-3) = l(n-3)$ is an $\Lp$-modulus of $f_\sigma$.
			Indeed:
			$h(n)\leq 2^{-n+2}$ since $x\mapsto (x+1)^{1/x}$ is decaying on the positive real line and takes value $2$ in $1$.
			Therefore, for an arbitrary shift $y$ and any $\sigma$
			\begin{equation*}\label{eq:difference}
				\|f_{\sigma} - \tau_y f_\sigma\|_p \leq 2\|f_\sigma\|_p < 2 h(n) \leq 2^{-n+3}.
			\end{equation*}
			Thus, for any $m<n-3$ zero is a valid value of an $\Lp$-modulus of $f_\sigma$.
			To see the statement for $m\geq n-3$ use Lemma~\ref{resu:classes with small lp-moduli}, which says that it suffices to estimate the $\Lp$-norm of the weak derivative of $f_\sigma$:
			\[ \|f'_\sigma\|_p \leq \frac{2h(n)}{w(n)} = 2^{l(n-3)-n+2} < 2^{l(n - 3)- n + 3}. \] 
			Thus, $m\mapsto m + l(n - 3)-n + 3$ is an $\Lp$-modulus of $f_\sigma$.
			
			Finally estimate the pairwise distance:
			The set $I$ was chosen such that whenever $\sigma\neq\sigma'$, then $\sigma$ and $\sigma'$ differ in at least $M=2^{l(n-3)-2}$ places.
			Thus
			\[ \|f_\sigma-f_{\sigma'}\|_p \geq 2h(n) \left(\frac{M w(n)}{p+1}\right)^{\frac 1p} = 2^{-n+2-\frac2p}\geq 2^{-n}. \]

			This proves the assertion.
		\end{proof}

	\subsection{Smoother approximations}
		The upper bound specified in Theorem~\ref{resu:fkT} contains an iteration of the $\Lp$-modulus while the lower bound does not.
		This leads to a huge gap between the upper and the lower bound for fast growing $\Lp$-moduli.
		In this chapter the upper bound is improved by introducing another representation that uses a sequence of approximating functions with improved regularity.

		Recall, that the convolution $h\convo f$ integrable functions $h,f$ is defined by
		\[ h\convo f := \int_{\RR^d} h(x-y)f(y) \dd y = \int_{\RR^d} h(y) f(x-y)\dd y. \]
		Recall the following well known result from the theory of convolution:
		\begin{prop}[derivatives]\label{resu:convolution differentiation formula}
			Whenever $f$ is integrable and $g$ is weakly differentiable, then $g\convo f$ is weakly differentiable and
			\[ \frac{\partial (g\convo f)}{\partial x_i} = \frac{\partial g}{\partial x_i} \convo f. \]
		\end{prop}
		Moreover, recall the following formula for the $\Lp$-norms of convoluted functions:
		\begin{prop}[norm estimate]\label{resu:convolution norm formula}
			Whenever $g\in \Lp$ where $1\leq p\leq \infty$ and $f\in\Lone$, then $g\convo f\in\Lp$ and
			\[ \norm{g\convo f}_{p} \leq \norm{g}_p \norm{f}_1. \]
		\end{prop}

		The sequence of continuous approximations $f_n$ from Definition~\ref{def:continuous approximations} can be understood to arise from the function $f$ by convoluting with the function sequence 
		\[ g_n:=2^{n}\chi_{[-2^{-n-1},2^{-n-1}]}. \quad\text{I.e.}\quad f_n= g_n\convo f. \]
		From this point of view Definition~\ref{d:xip} requires a $\xip$-name of a function $f$ to fulfill
		\[ \abs{\bin{\varphi(\str a,1^n)} - f_k(z)} < 2^{-n} \]
		whenever $\str a$ is an encoding of $\cell zk$.
		Furthermore, it is possible to translate between an encoding of $\cell zk$ and an the pair $\langle z,1^k\rangle$ in polynomial time.

		Let $\abs\cdot_\infty$ denote the supremum norm on $\RR^d$.
		Replacing the sequence $g_n$ with the following mollifier sequence lifts the approximations from being continuous to being weakly differentiable (thus the $D$):
		\vspace{.1cm}

		\noindent
		\begin{minipage}{.52\textwidth}
			\begin{defi}\label{def:mollifier functions}
				Define the \demph{mollifier sequence} $(g^D_m)_{m\in\NN}$ of functions $g^D_m:\RR^d\to\RR$ by
				\[ g^D_0(x):=\max\left\{1-\abs x_\infty,0\right\} \]
				and
				\[ g^D_m(x):= d2^{d(m-1)} g^D_0(2^m x). \]
			\end{defi}

			The function $g^D_1$ is illustrated in \Cref{fig:g_1}.
			One easily verifies that the support of $g^D_m$ is the ball $\left[-2^{-m}, 2^{-m}\right]^d$ of radius $2^{-m}$ around zero in supremum norm, that for any $m$
			\[ \int_{\RR^d} g^D_m \dd\lambda = \int_{\left[-2^{-m}, 2^{-m}\right]^d} g^D_m d\lambda = 1, \]
			\vspace{.1cm}
		\end{minipage}
		\begin{minipage}{.5\textwidth}
			\centering
			\begin{tikzpicture}
				\begin{axis}[grid=major,axis lines=middle,inner axis line style={=>},ticks=none,samples=45,view={30}{40}]
					\addplot3[surf,shader=interp,domain=-1.5:1.5] {max(0,1-max(abs(x+y),abs(y-x)))};
					\addplot3[domain=1:0] ({x-1},{0},{x});
					\addplot3[domain=1:0] ({0},{x-1},{x});
					\addplot3[domain=1:0] ({1-x},{0},{x});
					\addplot3[domain=1:0] ({x-1},{-x},{0});
					\addplot3[domain=1:0] ({1-x},{-x},{0});						
				\end{axis}
			\end{tikzpicture}
			
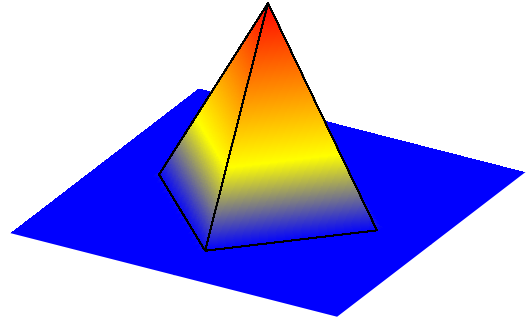
\captionof{figure}{The function $g^D_1$ for $d=2$.}\label{fig:g_1}
		\end{minipage}
		and that $g^D_0$ is weakly (partially) differentiable and the weak derivatives are essentially bounded with
		\[ \norm{\frac{d g^D_0}{dx_i}}_\infty= 1. \]
		Furthermore, if $\grad g^D_0$ denotes the vector of partial derivatives and $\norm\cdot_\infty$ the maximum of the supremum norms of the components of a vector, then for the gradient of $g^D_m$
		\[ \|\grad g^D_m\|_\infty = d 2^{d(m-1)+m}. \]

		\begin{defi}\label{def:differentiable approximations}
			For a function $f\in \Lp$ define the sequence $(f^D_m)_{m\in\NN}$ of \demph{differentiable approximations} $f^D_m:\RR^d\to\RR$ by
			\[ f^D_m:= g_m\star \tilde f= \int_{\RR^d} g^D_m(y) \tilde f(x-y)\dd y. \]
		\end{defi}

		The following representation replaces the information about the continuous approximations provided by $\xi_p$ by less information about smoother approximations.
		\begin{defi}\label{d:xipd}
			Define a second-order representation $\xip^D$ of $\LLL p{\Omega}$:
			A length-monotone string function $\varphi$ is a $\xip^D$-name of a function $f\in \LLL p{\Omega}$ if and for all $n,k\in \NN$ and each string $\str a$ encoding a dyadic number
			\[ \abs{\bin{\varphi(\langle \str a,\sdone^k, \sdone^n\rangle)} - f^D_k(\bin{\str a})} < 2^{-n}, \]
			and $\flength \varphi$ is an $\Lp$-modulus of $f$.
		\end{defi}

		It is possible to specify a bound on the supremum norm of the gradient of the \lq differentiable approximations\rq\ $f^D_n$:		
		\begin{lem}[gradient estimate]\label{resu:gradient estimate}
			For any $f\in\Lp$ the functions $f^D_m$ are weakly differentiable with
			\[ \norm{\grad f^D_m}_\infty \leq d 2^{d(m-1)+m}\|f\|_1 \]
		\end{lem}

		\begin{proof}
			From the properties of the mollifiers $g_m$ following Definition~\ref{def:mollifier functions} and the formulas for the convolution from Propositions~\ref{resu:convolution differentiation formula} and \ref{resu:convolution norm formula} get 
			\[ \norm{\frac{\partial f^D_m}{\partial x_i}}_\infty = \norm{\frac{\partial g^D_m}{\partial x_i}\convo f}_\infty \leq\norm{\frac{\partial g^D_m}{\partial x_i}}_\infty \norm f_1 = d2^{d(m-1)+m}\|f\|_1. \]
			Taking the supremum over $i$ proves the assertion.
		\end{proof}

		How good an approximation $f^D_m$ is to $f$ can be read from an $\Lp$-modulus.
		The smoothness, however, comes at a price that depends on the dimension (compare to Lemma~\ref{resu:continuous approximations}).

		\begin{lem}[approximation]\label{resu:differentiable approximations}
			Let $\mu$ be an $\Lp$-modulus of $f$, then
			\[ \|\tilde f-f^D_{\mu(n)}\|_p < 2^{-n+d}. \]
		\end{lem}

		\begin{proof}
			Like in the proof of \ref{resu:modulus of the continuous approximations} use $\int_{\RR^d} g^D_m \dd\lambda = 1$ to see that
			\begin{align*}
				\|\tilde f-f^D_{m}\|_p^p 
				& \leq \int_{\RR^d} \left( \int_{\RR^d} |\tilde f(x)-\tilde f(x-y)| g^D_m(y) \dd y\right)^p \dd x.
			\end{align*}
			Next note that from $g^D_0\leq 1$ and $p \geq 1$ it follows that
			\begin{equation}\label{eq:powers of gm}
				\begin{split}
					\left(g^D_m(y)\lambda([-2^{-m},2^{-m}]^d)^{1-\frac1p}\right)^p & = 2^{dm(p-1)} d^p2^{dp(m-1)}g^D_0(2^my)^p\\
					& \leq d^p 2^{d (m-p)}g^D_0(2^my) = 2^{d(p-1)}g^D_m(y).
				\end{split}
			\end{equation}
			Using the version of Hölder's inequality from Corollary~\ref{cor:hoelder} conclude
			\begin{align*}
				\left(\int_{\RR^d}|\tilde f-\tau_y\tilde f| g^D_m \dd \lambda\right)^p & \leq \left(\lambda([-2^{-m},2^{-m}]^d)^{\frac 1q}\norm{(\tilde f-\tau_y\tilde f) g^D_m}_p\right)^p \\
				& \stackrel{(\ref{eq:powers of gm})}\leq 2^{d(p-1)} \int_{[-2^{-m},2^{-m}]^d} |\tilde f-\tau_y\tilde f|^p g^D_m \dd \lambda 
			\end{align*}
			Therefore applying Fubini leads to
			\begin{align*}
				\|f-f^D_{m}\|_p^p 
				& = 2^{d(p-1)}\int_{[-2^{-m},2^{-m}]^d}\|\tilde f - \tau_y \tilde f\|_p^p g^D_m(y)\dd y.
			\end{align*}
			From this the assertion follows by setting $m:=\mu(n)$ and using that $\mu$ is an $\Lp$-modulus.
		\end{proof}

		Now it is possible to prove the main result of this section:

		\begin{thm}\label{resu:exponential time computability of the norm}
			With respect to $\xip^D$ the norm of $\Lp(\Omega)$ is exponential-time computable relative to $p$.
		\end{thm}

		\begin{proof}
			Denote the constant one vector $(1,\ldots,1)$ by $e$.
			Let $A$ be such that the distance of $\Omega$ to the complement of the box $K:=[-2^Ae,2^Ae]$ of radius $2^A$ around zero is bigger than one.
			Specify an oracle Turing machine that computes the norm as follows:
			Given a $\xip^D$-name $\varphi$ of $f\in\LLL p{\Omega}$ as oracle and a precision requirement $\sdone^n$ the machine first computes
			\[ N:=\mu(n+d+1),\quad M:= d(N - 1) + N + A + \mu(0)+\lceil\lb(d)\rceil \]
			and $k:=n+(d+1)(\mu(N)+2)+\mu(0)$.
			Let $\DD_M$ be the set of $z\in\DD^d\cap K$ of the form $m2^{-M}$.
			For each $z\in\DD_M$ the machine chooses some binary encoding $\str a_z\in \albe^*$ (for example the unique canceled encoding such that the encoding of the enumerator has no leading zeros) and queries $\varphi$ for
			\[ d_{z,n}:= \bin{\varphi(\langle\str a_{z},\sdone^N,\sdone^{k}\rangle)}. \]
			After each query, it adds the $p$-th power of the result to the sum of the previous queries.
			In the end, the $p$-th root of the number is returned.

			To see that this machine returns a correct approximation, note that it returns the $\Lp$-norm of the function
			\[ F_{n}:= \sum_{z\in\DD_M} d_{z,n}\chi_{\cell zM}. \]
			To see that $F_n$ is a good approximation to $f$ in the $\Lp$ norm, write
			\begin{align}
				\|f-F_n\|_p & \leq \|\tilde f-f^D_{N}\|_p + \norm{f^D_N - F_n}_p
			\end{align}
			Each of the summands of the right hand side is smaller than $2^{-n-1}$:
			The first term is taken care of by the choice of $N$ together with the approximation property of the sequence $f^D_N$ from Lemma~\ref{resu:differentiable approximations}.
			For the second summand, note that each $x\in\Omega$ is $2^{-M}$ close to some $z\in K$ and that for these $z$
			\[ \abs{f^D_N(z)-F_{n}(z)} = \abs{f_N^D(z) - d_{z,n}} <2^{dN-k}. \]
			By choice of $M$, the gradient estimate of $f^D_N$ from Lemma~\ref{resu:gradient estimate} together with the bound on the $\Lp$-norm from Lemma~\ref{resu:norm bound from modulus} and since $F_{n}$ is piecewise constant it follows that
			\[ \norm{f^D_N-F_{n}}_p \leq\lambda(\Omega)^{\frac 1p} \norm{f^D_N - F_{n}}_\infty < 2^{-n-1}. \]
			Thus, from $\big|{\norm{f}_p-\norm{F_n}_p}\big| \leq \norm{f-F_n}_p$ it follows that the return value is indeed a valid $2^{-n}$ approximation to the norm of $f$.

			The assertion now follows from the fact that the machine carries out a loop that takes time linear in $\flength{\varphi}(n+C)$ (provided that approximations to $p$ are given) an exponential number of times.
		\end{proof}

		Using the content of \Cref{sec:connecting metric entropy and complexity}, the next result can be regarded as a corollary of the above. However, since it is also possible to give an independent proof without relying on that section we list it as a theorem.
		Recall the \fkT sets from Definition~\ref{def:fkT sets}.

		\begin{thm}\label{resu:improved upper fkT}
			There are constants $A,B,C\in\NN$ such that 
			\[ \size{K^p_l}(n)\leq 2^{A l(n+d+1) + l(0)+B} + n + C. \]
		\end{thm}

		\begin{proof}
			We give two sketches:
			On the one hand it is possible to follow the proof of Theorem~\ref{resu:fkT} and replace the lemmas from \Cref{sec:sub:equivalence to the standard representation} by the lemmas of this section.
			On the other hand one can add an oracle for $p$ to each $\Lp$-name of a function and in this way construct a representation such that applying the first item of Theorem~\ref{resu:metric entropy and complexity} proves the claim.
		\end{proof}
		While the first sketch relies on the quantitative version of the \aaT Theorem~\ref{resu:aaT} and is therefore bound to the unit interval, the second sketch is not bound to such a restrictive setting:
		It remains correct for more general domains and higher dimensions.

\section*{Conclusion}
	
	A comprehensive summary of the content of the paper can be found in the introduction.
	Thus, this conclusion concentrates on high level comments and mentioning additional results or improvements of results and points out some loose ends.
	The remarks follow the general outline of the paper.

	Like for the standard representation of continuous functions, the minimality property of the singular representation from Theorem~\ref{resu:minimality of the singular representation d} also applies to arbitrary range-restrictions.
	The discontinuity of the singular representation can be strengthened:
	Theorem~\ref{resu:discontinuity norm d} remains true if the norm topology is replaced by the weak topology.
	The singular representation seems to have a straight forward generalization to locally integrable functions.
	On locally integrable functions the usual norms do not make sense anymore and different topologies are considered.
	It would be interesting to find out whether the topology of the singular representation coincides with one previously considered by analysts.

	Some might argue that the choice of an integration operator is too restrictive.
	At least in higher dimensions the restriction of the integral operator to only integrate over boxes seems very severe.
	This restrictive setting, however, seems unavoidable.
	Polynomial-time computability of many possible extensions is ruled out by hardness results proven in \cite{MR1137517}.
	The same holds for the approach to consider a function to be a functional on the continuous functions.

	For the definition of the $\Lp$-modulus in Definition~\ref{def:lp-moduli} the function was extended to the whole space by zero before integrating.
	This is a convention: One could instead have integrated over the intersection of the domains of the function and the shifted function.
	However, in this case the property that exactly the $\Lp$-functions allow a modulus is lost.
	If one uses this modification to define a representation technical difficulties are encountered when a proof of equivalence to the Cauchy representation is attempted.
	For complexity considerations it seems impossible to progress on this path without restricting the domains.

	The first part of the proof of Theorem~\ref{resu:equivalence to the standard representation} can be seen to show the stronger statement of computable openness of the representation as introduced in \cite{kihara2014point}.
	Furthermore, if exponential-time computability is introduced to allow a full second-order polynomial in the exponent (in contrast to Definition~\ref{def:exponential time computability}, where no function argument iteration is allowed in the exponent), the second part of the proof on page \pageref{proof:second part of equivalence} shows exponential-time translatability to the Cauchy representation.
	From this a weaker form of exponential-time computability of the norm than that from Theorem~\ref{resu:exponential time computability of the norm} follows.

	Recall from the introduction, that in practice maximization is considered difficult while integration is considered feasible.
	This is reflected in the second-order representations introduced in this paper:
	Neither the representation $\xis$ nor the representations $\xip$ allow in a straight-forward way to maximize a continuous or smooth  function.
	Indeed\footnote{thanks to Akitoshi Kawamura for pointing this out.}, modifications of the smooth functions considered by Ko and Friedman in \cite{MR666209} show that the maximization operator will not preserve polynomial-time computability with respect to these representations unless $\p=\np$.

	While ordinary differential equations are a field of application for Sobolev spaces (compare for instance \cite[chapter 8.4]{MR2759829}), partial differential equations are by far the most important application.
	However, many of the arguments from \Cref{sec:sobolev spaces} cannot be translated in the most straight forward way to higher dimensions.
	For instance: Existence of a weak derivative does not imply continuity in higher dimensions.
	For the inclusions to make sense in higher dimensions further assumptions are necessary.
	Even if these assumptions are met, Theorem~\ref{resu:Lp-modulus of derivative to modulus of continuity} cannot be straight forwardly replaced.
	Indeed, the argument from Proposition~\ref{resu: from lp-modulus to singularity modulus} cannot carry over to higher dimensions in the straight forward way, as it would only mention derivatives of first order and it is known that existence of the first weak partial derivatives does not imply continuity.

	Partial differential equations have received increased interest in computable analysis in the last years.
	Compare for instance \cite{MR1694445,MR2351947,MR2275415,MR2275411}.
	There is a plethora of results for solving partial differential equations from numerical analysis.
	It seems reasonable to assume that formulating these algorithms in a rigorous framework and lifting results from the references above to a complexity theoretical level should be closely connected tasks.
	For instance many of the results from \cite{MR2275411} are interesting in one dimension already.

	All results from this paper that mention exponential-time computability can be improved to use polynomial-space computability instead.
	The model of space bounded computation in presence of oracles, however, is not completely straight forward:
	The right model of oracle access is a stack of finite depth (compare \cite{MR3259646,MR973445}).
	
	The representations $\xip$ and the representation $\xip^D$ from the last chapter can be combined to a representation featuring both polynomial-time computability of integrals and exponential-time computability of the norm.
	However, it does not seem reasonable to add the information provided by $\xip^D$:
	It increases the amount of information that has to be provided to specify a function for the sake of improving the runtime of an exponential-time computable (so not feasible) operation on input of big $\Lp$-modulus.
	The first sketch of a proof of Theorem~\ref{resu:improved upper fkT} suggests that convoluting with even smoother functions does not lead to further improvements in performance:
	The dominant term in the running time is independent of smoothness:
	The a supremum norm estimate obtained from Proposition~\ref{resu:convolution norm formula}.

	Classification theorems for the compact subsets of function spaces are of interest to analysts and approximation theorists for reasons independent of those sketched in this paper.
	They have been investigated for a long time and are well developed.
	The link between quantitative versions of these results and optimality results for running times provides a rich resource for finding interesting representations.
	Such results are in particular known for Banach space valued
        functions.

	\bibliography{bib}{}

\begin{thebibliography}{KMRZ15}

\bibitem[AB09]{MR2500087}
Sanjeev Arora and Boaz Barak.
\newblock {\em Computational complexity}.
\newblock Cambridge University Press, Cambridge, 2009.
\newblock A modern approach.
\newblock \href {http://dx.doi.org/10.1017/CBO9780511804090}
  {\path{doi:10.1017/CBO9780511804090}}.

\bibitem[BB85]{MR804042}
Errett Bishop and Douglas Bridges.
\newblock {\em Constructive analysis}, volume 279 of {\em Grundlehren der
  Mathematischen Wissenschaften}.
\newblock Springer-Verlag, Berlin, 1985.
\newblock \href {http://dx.doi.org/10.1007/978-3-642-61667-9}
  {\path{doi:10.1007/978-3-642-61667-9}}.

\bibitem[BH02]{MR1923900}
Vasco Brattka and Peter Hertling.
\newblock Topological properties of real number representations.
\newblock {\em Theoret. Comput. Sci.}, 284(2):241--257, 2002.
\newblock Computability and complexity in analysis (Castle Dagstuhl, 1999).
\newblock \href {http://dx.doi.org/10.1016/S0304-3975(01)00066-4}
  {\path{doi:10.1016/S0304-3975(01)00066-4}}.

\bibitem[BK02]{MR1911553}
Samuel~R. Buss and Bruce~M. Kapron.
\newblock Resource-bounded continuity and sequentiality for type-two
  functionals.
\newblock {\em ACM Trans. Comput. Log.}, 3(3):402--417, 2002.
\newblock Special issue on logic in computer science (Santa Barbara, CA, 2000).
\newblock \href {http://dx.doi.org/10.1145/507382.507387}
  {\path{doi:10.1145/507382.507387}}.

\bibitem[Bre11]{MR2759829}
Haim Brezis.
\newblock {\em Functional analysis, {S}obolev spaces and partial differential
  equations}.
\newblock Universitext. Springer, New York, 2011.
\newblock \href {http://dx.doi.org/10.1007/978-0-387-70914-7}
  {\path{doi:10.1007/978-0-387-70914-7}}.

\bibitem[Bus88]{MR973445}
Jonathan~F. Buss.
\newblock Relativized alternation and space-bounded computation.
\newblock {\em J. Comput. System Sci.}, 36(3):351--378, 1988.
\newblock Structure in Complexity Theory Conference (Berkeley, CA, 1986).
\newblock \href {http://dx.doi.org/10.1016/0022-0000(88)90034-7}
  {\path{doi:10.1016/0022-0000(88)90034-7}}.

\bibitem[BY06]{MR2275411}
Vasco Brattka and Atsushi Yoshikawa.
\newblock Towards computability of elliptic boundary value problems in
  variational formulation.
\newblock {\em J. Complexity}, 22(6):858--880, 2006.
\newblock \href {http://dx.doi.org/10.1016/j.jco.2006.04.007}
  {\path{doi:10.1016/j.jco.2006.04.007}}.

\bibitem[FGH14]{MR3239272}
Hugo F{\'e}r{\'e}e, Walid Gomaa, and Mathieu Hoyrup.
\newblock Analytical properties of resource-bounded real functionals.
\newblock {\em J. Complexity}, 30(5):647--671, 2014.
\newblock \href {http://dx.doi.org/10.1016/j.jco.2014.02.008}
  {\path{doi:10.1016/j.jco.2014.02.008}}.

\bibitem[Fri84]{MR748898}
Harvey Friedman.
\newblock The computational complexity of maximization and integration.
\newblock {\em Adv. in Math.}, 53(1):80--98, 1984.
\newblock \href {http://dx.doi.org/10.1016/0001-8708(84)90019-7}
  {\path{doi:10.1016/0001-8708(84)90019-7}}.

\bibitem[Grz55]{MR0086756}
A.~Grzegorczyk.
\newblock Computable functionals.
\newblock {\em Fund. Math.}, 42:168--202, 1955.

\bibitem[IRK01]{MR1826285}
Robert~J. Irwin, James~S. Royer, and Bruce~M. Kapron.
\newblock On characterizations of the basic feasible functionals. {I}.
\newblock {\em J. Funct. Programming}, 11(1):117--153, 2001.
\newblock Special issue on functional programming and computational complexity
  (Baltimore, MD, 1998).
\newblock \href {http://dx.doi.org/10.1017/S0956796800003841}
  {\path{doi:10.1017/S0956796800003841}}.

\bibitem[KC96]{MR1374053}
B.~M. Kapron and S.~A. Cook.
\newblock A new characterization of type-{$2$} feasibility.
\newblock {\em SIAM J. Comput.}, 25(1):117--132, 1996.
\newblock \href {http://dx.doi.org/10.1137/S0097539794263452}
  {\path{doi:10.1137/S0097539794263452}}.

\bibitem[KC10]{MR2743298}
Akitoshi Kawamura and Stephen Cook.
\newblock Complexity theory for operators in analysis.
\newblock In {\em S{TOC}'10---{P}roceedings of the 2010 {ACM} {I}nternational
  {S}ymposium on {T}heory of {C}omputing}, pages 495--502. ACM, New York, 2010.
\newblock \href {http://dx.doi.org/10.1145/1806689.1806758}
  {\path{doi:10.1145/1806689.1806758}}.

\bibitem[KF82]{MR666209}
Ker-I Ko and Harvey Friedman.
\newblock Computational complexity of real functions.
\newblock {\em Theoret. Comput. Sci.}, 20(3):323--352, 1982.
\newblock \href {http://dx.doi.org/10.1016/S0304-3975(82)80003-0}
  {\path{doi:10.1016/S0304-3975(82)80003-0}}.

\bibitem[Kle52]{MR0051790}
Stephen~Cole Kleene.
\newblock {\em Introduction to metamathematics}.
\newblock D. Van Nostrand Co., Inc., New York, N. Y., 1952.

\bibitem[KM82]{MR717246}
G.~Kreisel and A.~Macintyre.
\newblock Constructive logic versus algebraization. {I}.
\newblock In {\em The {L}. {E}. {J}. {B}rouwer {C}entenary {S}ymposium
  ({N}oordwijkerhout, 1981)}, volume 110 of {\em Stud. Logic Found. Math.},
  pages 217--260. North-Holland, Amsterdam, 1982.
\newblock \href {http://dx.doi.org/10.1016/S0049-237X(09)70130-2}
  {\path{doi:10.1016/S0049-237X(09)70130-2}}.

\bibitem[KMRZ15]{MR3377508}
Akitoshi Kawamura, Norbert M{\"u}ller, Carsten R{\"o}snick, and Martin Ziegler.
\newblock Computational benefit of smoothness: {P}arameterized bit-complexity
  of numerical operators on analytic functions and {G}evrey's hierarchy.
\newblock {\em J. Complexity}, 31(5):689--714, 2015.
\newblock \href {http://dx.doi.org/10.1016/j.jco.2015.05.001}
  {\path{doi:10.1016/j.jco.2015.05.001}}.

\bibitem[Ko91]{MR1137517}
Ker-I Ko.
\newblock {\em Complexity theory of real functions}.
\newblock Progress in Theoretical Computer Science. Birkh\"auser Boston, Inc.,
  Boston, MA, 1991.
\newblock \href {http://dx.doi.org/10.1007/978-1-4684-6802-1}
  {\path{doi:10.1007/978-1-4684-6802-1}}.

\bibitem[KO14]{MR3259646}
Akitoshi Kawamura and Hiroyuki Ota.
\newblock Small complexity classes for computable analysis.
\newblock In {\em Mathematical foundations of computer science 2014. {P}art
  {II}}, volume 8635 of {\em Lecture Notes in Comput. Sci.}, pages 432--444.
  Springer, Heidelberg, 2014.
\newblock \href {http://dx.doi.org/10.1007/978-3-662-44465-8_37}
  {\path{doi:10.1007/978-3-662-44465-8_37}}.

\bibitem[Koh96]{MR1462200}
Ulrich Kohlenbach.
\newblock Mathematically strong subsystems of analysis with low rate of growth
  of provably recursive functionals.
\newblock {\em Arch. Math. Logic}, 36(1):31--71, 1996.
\newblock \href {http://dx.doi.org/10.1007/s001530050055}
  {\path{doi:10.1007/s001530050055}}.

\bibitem[Koh05]{MR2130066}
Ulrich Kohlenbach.
\newblock Some computational aspects of metric fixed-point theory.
\newblock {\em Nonlinear Anal.}, 61(5):823--837, 2005.
\newblock \href {http://dx.doi.org/10.1016/j.na.2005.01.075}
  {\path{doi:10.1016/j.na.2005.01.075}}.

\bibitem[KP14a]{MR3219039}
Akitoshi Kawamura and Arno Pauly.
\newblock Function spaces for second-order polynomial time.
\newblock In {\em Language, life, limits}, volume 8493 of {\em Lecture Notes in
  Comput. Sci.}, pages 245--254. Springer, Cham, 2014.
\newblock \href {http://dx.doi.org/10.1007/978-3-319-08019-2_25}
  {\path{doi:10.1007/978-3-319-08019-2_25}}.

\bibitem[KP14b]{kihara2014point}
Takayuki Kihara and Arno Pauly.
\newblock Point degree spectra of represented spaces.
\newblock {\em arXiv preprint
  \href{https://arxiv.org/abs/1405.6866}{arXiv:1405.6866}}, 2014.

\bibitem[KSZ16]{CIE2016}
Akitoshi Kawamura, Florian Steinberg, and Martin Ziegler.
\newblock Towards computational complexity theory on advanced function spaces
  in analysis.
\newblock In {\em Pursuit of the Universal: 12th Conference on Computability in
  Europe}, pages 142--152. Springer International Publishing, 2016.
\newblock \href {http://dx.doi.org/10.1007/978-3-319-40189-8_15}
  {\path{doi:10.1007/978-3-319-40189-8_15}}.

\bibitem[KT59]{MR0112032}
A.~N. Kolmogorov and V.~M. Tihomirov.
\newblock {$\varepsilon $}-entropy and {$\varepsilon $}-capacity of sets in
  function spaces.
\newblock {\em Uspehi Mat. Nauk}, 14(2 (86)):3--86, 1959.

\bibitem[Lac55]{MR0072080}
Daniel Lacombe.
\newblock Extension de la notion de fonction r\'ecursive aux fonctions d'une ou
  plusieurs variables r\'eelles. {II}, {III}.
\newblock {\em C. R. Acad. Sci. Paris}, 241:13--14, 151--153, 1955.

\bibitem[Lam06]{MR2275414}
Branimir Lambov.
\newblock The basic feasible functionals in computable analysis.
\newblock {\em J. Complexity}, 22(6):909--917, 2006.
\newblock \href {http://dx.doi.org/10.1016/j.jco.2006.06.005}
  {\path{doi:10.1016/j.jco.2006.06.005}}.

\bibitem[LLM01]{MR1795248}
S.~Labhalla, H.~Lombardi, and E.~Moutai.
\newblock Espaces m\'etriques rationnellement pr\'esent\'es et complexit\'e: le
  cas de l'espace des fonctions r\'eelles uniform\'ement continues sur un
  intervalle compact.
\newblock {\em Theoret. Comput. Sci.}, 250(1-2):265--332, 2001.
\newblock \href {http://dx.doi.org/10.1016/S0304-3975(99)00139-5}
  {\path{doi:10.1016/S0304-3975(99)00139-5}}.

\bibitem[Lon05]{MR2143877}
John~R. Longley.
\newblock Notions of computability at higher types. {I}.
\newblock In {\em Logic {C}olloquium 2000}, volume~19 of {\em Lect. Notes
  Log.}, pages 32--142. Assoc. Symbol. Logic, Urbana, IL, 2005.

\bibitem[Lor66]{lorentz1966}
G.~G. Lorentz.
\newblock Metric entropy and approximation.
\newblock {\em Bull. Amer. Math. Soc.}, 72(6):903--937, 11 1966.
\newblock \href {http://dx.doi.org/10.1090/S0002-9904-1966-11586-0}
  {\path{doi:10.1090/S0002-9904-1966-11586-0}}.

\bibitem[Meh76]{MR0411947}
Kurt Mehlhorn.
\newblock Polynomial and abstract subrecursive classes.
\newblock {\em J. Comput. System Sci.}, 12(2):147--178, 1976.
\newblock Sixth Annual ACM Symposium on the Theory of Computing (Seattle,
  Wash., 1974).
\newblock \href {http://dx.doi.org/10.1145/800119.803890}
  {\path{doi:10.1145/800119.803890}}.

\bibitem[Mun00]{munkres2000topology}
J.R. Munkres.
\newblock {\em Topology}.
\newblock Featured Titles for Topology Series. Prentice Hall, Incorporated,
  Upper Saddle River, NJ 07458, 2000.

\bibitem[PER89]{MR1005942}
Marian~B. Pour-El and J.~Ian Richards.
\newblock {\em Computability in analysis and physics}.
\newblock Perspectives in Mathematical Logic. Springer-Verlag, Berlin, 1989.

\bibitem[Rud87]{MR924157}
Walter Rudin.
\newblock {\em Real and complex analysis}.
\newblock McGraw-Hill Book Co., New York, third edition, 1987.

\bibitem[Sch02a]{SchroederPhD}
Matthias Schr{\"o}der.
\newblock {\em Admissible Representations for Continuous Computations}.
\newblock PhD thesis, FernUniversität Hagen, 2002.

\bibitem[Sch02b]{MR1923914}
Matthias Schr{\"o}der.
\newblock Extended admissibility.
\newblock {\em Theoret. Comput. Sci.}, 284(2):519--538, 2002.
\newblock Computability and complexity in analysis (Castle Dagstuhl, 1999).
\newblock \href {http://dx.doi.org/10.1016/S0304-3975(01)00109-8}
  {\path{doi:10.1016/S0304-3975(01)00109-8}}.

\bibitem[Sch04]{MR2090390}
Matthias Schr{\"o}der.
\newblock Spaces allowing type-2 complexity theory revisited.
\newblock {\em MLQ Math. Log. Q.}, 50(4-5):443--459, 2004.
\newblock \href {http://dx.doi.org/10.1002/malq.200310111}
  {\path{doi:10.1002/malq.200310111}}.

\bibitem[Ste16]{SteinbergPhD}
Florian Steinberg.
\newblock {\em Computational Complexity Theory for Advanced Function Spaces in
  Analysis}.
\newblock PhD thesis, Technische Universität Darmstadt, 2016.

\bibitem[Sud01]{MR1948693}
Madhu Sudan.
\newblock Coding theory: tutorial \& survey.
\newblock In {\em 42nd {IEEE} {S}ymposium on {F}oundations of {C}omputer
  {S}cience ({L}as {V}egas, {NV}, 2001)}, pages 36--53. IEEE Computer Soc., Los
  Alamitos, CA, 2001.

\bibitem[Tim94]{MR1262128}
A.~F. Timan.
\newblock {\em Theory of approximation of functions of a real variable}.
\newblock Dover Publications, Inc., New York, 1994.
\newblock Translated from the Russian by J. Berry, Translation edited and with
  a preface by J. Cossar, Reprint of the 1963 English translation.

\bibitem[Tur36]{turing1936computable}
Alan~Mathison Turing.
\newblock On computable numbers, with an application to the
  entscheidungsproblem.
\newblock {\em J. of Math}, 58(345-363):5, 1936.
\newblock \href {http://dx.doi.org/10.2307/2268810}
  {\path{doi:10.2307/2268810}}.

\bibitem[Wei00]{MR1795407}
Klaus Weihrauch.
\newblock {\em Computable analysis}.
\newblock Texts in Theoretical Computer Science. An EATCS Series.
  Springer-Verlag, Berlin, 2000.
\newblock An introduction.
\newblock \href {http://dx.doi.org/10.1007/978-3-642-56999-9}
  {\path{doi:10.1007/978-3-642-56999-9}}.

\bibitem[Wei03]{MR1952428}
Klaus Weihrauch.
\newblock Computational complexity on computable metric spaces.
\newblock {\em MLQ Math. Log. Q.}, 49(1):3--21, 2003.
\newblock \href {http://dx.doi.org/10.1002/malq.200310001}
  {\path{doi:10.1002/malq.200310001}}.

\bibitem[Wer00]{MR1787146}
Dirk Werner.
\newblock {\em Funktionalanalysis}.
\newblock Springer-Verlag, Berlin, extended edition, 2000.

\bibitem[WZ06]{MR2275415}
Klaus Weihrauch and Ning Zhong.
\newblock Computing {S}chr\"odinger propagators on type-2 {T}uring machines.
\newblock {\em J. Complexity}, 22(6):918--935, 2006.
\newblock \href {http://dx.doi.org/10.1016/j.jco.2006.06.001}
  {\path{doi:10.1016/j.jco.2006.06.001}}.

\bibitem[WZ07]{MR2351947}
Klaus Weihrauch and Ning Zhong.
\newblock Computable analysis of the abstract {C}auchy problem in a {B}anach
  space and its applications. {I}.
\newblock {\em MLQ Math. Log. Q.}, 53(4-5):511--531, 2007.
\newblock \href {http://dx.doi.org/10.1002/malq.200710015}
  {\path{doi:10.1002/malq.200710015}}.

\bibitem[Zho99]{MR1694445}
Ning Zhong.
\newblock Computability structure of the {S}obolev spaces and its applications.
\newblock {\em Theoret. Comput. Sci.}, 219(1-2):487--510, 1999.
\newblock Computability and complexity in analysis (Castle Dagstuhl, 1997).
\newblock \href {http://dx.doi.org/10.1016/S0304-3975(98)00302-8}
  {\path{doi:10.1016/S0304-3975(98)00302-8}}.

\bibitem[ZZ99]{MR1724414}
Ning Zhong and Bing-Yu Zhang.
\newblock {$L^p$}-computability.
\newblock {\em MLQ Math. Log. Q.}, 45(4):449--456, 1999.
\newblock \href {http://dx.doi.org/10.1002/malq.19990450403}
  {\path{doi:10.1002/malq.19990450403}}.

\end{thebibliography}
	\bibliographystyle{alphaurl}
	\newpage
\end{document}